\documentclass[smallextended]{svjour3}       

\usepackage[top=2cm, bottom=2.5cm, left=2cm, right=2cm]{geometry} 
\usepackage[justification=centering]{caption}
\usepackage{color}

\usepackage{amsmath} 
\usepackage{proof} 
\usepackage{mathrsfs} 

\usepackage[ruled,linesnumbered]{algorithm2e}
\usepackage{arydshln} 

\usepackage[colorlinks=true,linkcolor=blue,citecolor=blue,pdfborder={0 0 0}]{hyperref}             

\usepackage{graphicx}     
\usepackage{subcaption}   
\graphicspath{{image/}}   
\usepackage{float}        

\usepackage{multirow}     

\usepackage{natbib}
\usepackage{booktabs}     

\usepackage{pythonhighlight} 
\begin{document}

\title{Confounding  analysis of  $s$-level designs with multi-block variables}
\subtitle{}
\author{Wenbo Hu        \and    Zhiming Li*  }

\institute{Wenbo Hu,
              \email{huwenbo@stu.xju.edu.cn}           
           \\
           Zhiming Li, 
              \email{zmli@xju.edu.cn}  \\         
           College of Mathematics and System Science, Xinjiang University,  Urumqi 830046, China
}

\date{Received: date / Accepted: date}

\maketitle
\begin{abstract}
In practical experiments, block variables often arise from multiple sources of heterogeneity. To address the confounding problem, this paper proposes a blocked aliased component-number pattern (B${}^2$-ACNP) to analyze the confounding properties of $s$-level designs with multi-block variables.  We calculate the values of   (B${}^2$-ACNP) via a blocked wordlength distribution matrix. The classification patterns of existing criteria can be expressed as functions of specific elements within the  B${}^2$-ACNP, thereby establishing connections within a unified framework. Further, we provide confounding algorithms and visualization methods of the B${}^2$-ACNP. Finally, case analysis clarifies the significant role of the B${}^2$-ACNP.  The Python code is available in the Appendix.
	
\keywords{ Blocked design  \and Blocked aliased component-number pattern \and Blocked wordlength distribution matrix}
\end{abstract}

\section{Introduction}\label{sec1} 
In fractional factorial (FF) designs, large-scale experimental units often exhibit heterogeneity. Blocking experimental units can effectively control the heterogeneity.  Bisgaard (\citeyear{bisgaard_note_1994}) classified block problems in FF designs into two types: single-block and multi-block variables. For the former, scholars have developed various criteria to select optimal designs. The widely used one is the blocked minimum aberration (MA-type) criterion via the wordlength pattern of the treatment and block defining contrast subgroup (Sitter et al., \citeyear{sitter_fractional_1997}; Chen and Cheng, \citeyear{chen1999theory}; Zhang and Park, \citeyear{zhang2000optimal}; Cheng and Wu, \citeyear{cheng2002choice}). However, the criterion cannot distinguish between blocked designs with the same wordlength pattern. An alternative method is to choose optimal blocked designs by maximizing the number of clear main effects and clear two-factor interaction effects (2fis) (Chen et al., \citeyear{chen_results_2006}; Ai and He,  \citeyear{ai_efficient_2006}). The method is applicable when experimenters focus only on main effects and 2fis, and is invalid for blocked designs with no clear effects or all clear effects. In modeling,  the maximum estimation capacity (MEC) criterion compares the number of models with estimable main effects and important 2fis in blocked designs (Mukerjee and Cheng, \citeyear{mukerjee_blocked_2001}; Cheng and Mukerjee, \citeyear{cheng_regular-fractional_2003}). However, the MEC criterion also struggles to screen designs with high resolution, assuming that the remaining 2fis are negligible. Note that the classification patterns of these criteria are closely related to the confounding among the various effects. Zhang and Mukerjee (\citeyear{zhang_general_2009}) proposed a blocked aliased component-number pattern (B-ACNP) for $s$-level blocked designs, which accounts for confounding information more rigorously, and the aforementioned patterns can all be expressed as functions of some elements in the B-ACNP.  Under the B-ACNP, a general minimum lower-order confounding (B-GMC) criterion is considered and applied to prior information about the treatment factors. Until now, there have been many achievements in the construction of two-level B-GMC designs (Tan and Zhang, \citeyear{tan_construction_2013}; Zhao et al., \citeyear{zhao_note_2016}; Zhao and Zhao, \citeyear{zhao_results_2018}).  Wei et al. (\citeyear{wei_blocked_2014}) further refined the B-ACNP, called the B${}^1$-ACNP, and proposed the B${}^1$-GMC criterion for two-level blocked designs. The lower-confounding properties and B${}^1$-GMC construction of two-level blocked designs can be found in Zhao et al. (\citeyear{zhao_construction_2013}), Guo et al. (\citeyear{GUO201598}), Zhao et al. (\citeyear{zhao_theory_2016}),  Zhao and Sun (\citeyear{zhao_constructing_2017}), and Wang et al. (\citeyear{wang_blocked_2017}). Wang et al. (\citeyear{wang_three-level_2020}) extended the B${}^1$-ACNP to three-level cases, and Li et al. (\citeyear{li_results_2025}) further explored the construction of B${}^1$-GMC designs at three levels. More recently, Li and Li (\citeyear{li_general_2025}) generalized the B${}^1$-GMC criterion to \(s\)-level designs, significantly broadening its scope of application.

Previous studies have primarily focused on designs with a single block variable. However, in practical experiments, issues related to multi-block variables have attracted increasing attention. Bisgaard (\citeyear{bisgaard_note_1994}) pointed out that the soil may have a heterogeneous effect on both column and row variables simultaneously in a rectangular experimental field.  Williams et al. (\citeyear{williams_experimental_2021}) conducted an experiment on legume signal peptides that included seven block variables. This demonstrates that a systematic study of designs with multi-block variables holds significant theoretical importance and practical necessity. Since a single-block variable is a special case of a multi-block variable, the criteria for single-block variables often fail to select optimal designs for multi-block variables (Zhang et al., \citeyear{zhang_optimal_2011}; Zhao et al., \citeyear{zhao_construction_2017}). Although the MA, CE, and MEC criteria can be extended to the multi-block variable case (Zhao and Zhao, \citeyear{zhao_mixed-level_2019}, \citeyear{zhao_results_2019}; Zhao and Zhao, \citeyear{zhao_minimum_2021}), as in the single-block case, these criteria still suffer from analogous shortcomings in experiments with multi-block variables. For the two-level multi-block variable scenario, Zhang et al. (\citeyear{zhang_optimal_2011}) classified treatment effects into four categories:  \(g\)-class, \(b\)-class, \(m\)-class, and   \(\phi\)-class. Based on this, they proposed  \(\text{B}^2\text{-GMC}\) criterion via a four-category \(\text{B}^2\text{-AENP}\). The results show that the MA, CE, and MEC criteria can also be uniformly expressed by the \(\text{B}^2\text{-AENP}\). The properties of two-level \(\text{B}^2\text{-GMC}\) designs have been studied in Zhao et al. (\citeyear{zhao_construction_2017}) and Zhao (\citeyear{zhao_general_2021}).

Although the \(\text{B}^2\text{-AENP}\) and \(\text{B}^2\text{-GMC}\) criterion have made progress, there are still two aspects that need improvement: (a) most of the above research pays close attention to the two-level  \(\text{B}^2\text{-GMC}\) criterion. However, there are fewer research findings on higher-level  \(\text{B}^2\text{-AENP}\), which is the basis of the \(\text{B}^2\text{-GMC}\) criterion.   (b)  Existing literature on the calculation of the \(\text{B}^2\text{-AENP}\) is typically limited to main effects and two-factor interactions. When the resolution of a design with multi-block variables reaches VII, all its main effects and two- and three-factor interactions are clear. This necessitates paying attention to the potential influence of higher-order interaction effects. Moreover, aliasing analysis for higher-order effects is cumbersome and computationally intensive. Based on the above analysis, this paper extends the \(\text{B}^2\text{-AENP}\) to higher-level cases and calculates the \(\text{B}^2\text{-AENP}\) via a series of algorithms. The main contributions and novelties are summarized as follows:

(i) The \(\text{B}^2\text{-ACNP}\) is extended to \(s\)-level designs with multi-block variables, where \(s\) is a prime or a prime power, thereby broadening the scope of this criterion. Construction properties of the B\(^2\)-GMC designs are studied via complementary sets.

(ii) A new method is proposed for computing aliasing among effects of various orders using the blocked wordlength distribution matrix (B-WDM). Compared to the existing approach of Zhang and Mukerjee (\citeyear{zhang_general_2009}), the proposed method is more suitable for calculating aliasing relationships among high-order effects and for obtaining the complete \(\text{B}^2\text{-ACNP}\).

(iii) The relationships between the \(\text{B}^2\text{-ACNP}\) and other criterion classification patterns are analyzed. It is found that the MA, CE, and MEC criteria can all be derived from the \(\text{B}^2\text{-ACNP}\) via the B-WDM.

(iv) The algorithms and visualization are provided to calculate and illustrate the values of \(\text{B}^2\text{-ACNP}\), and the empirical analysis demonstrates the application. 

The remainder of this paper is organized as follows. Section \ref{sec2} presents the definitions of the \(\text{B}^2\text{-ACNP}\), the \(\text{B}^2\text{-GMC}\) criterion, and the B-WDM for \(s\)-level designs with multi-block variables, and discusses some properties of the aliasing components. In Section \ref{sec3}, we systematically explore the connections between the \(\text{B}^2\text{-ACNP}\) and the classification patterns of the MA, CE, and MEC criteria. Section \ref{sec4} provides a computer algorithm for the \(\text{B}^2\text{-ACNP}\) and a visual analysis of the aliasing patterns. Section \ref{sec5} demonstrates the practicality of the proposed criterion and algorithms through two illustrative examples. Finally, a brief conclusion and future research are given in Section \ref{sec6}. All algorithms in this paper are implemented in Python.
\section{Blocked aliased component-number pattern}\label{sec2}

Following the works of Zhang et al. (\citeyear{zhang_optimal_2011}) and Zhao et al. (\citeyear{zhao_construction_2017}, \citeyear{zhao_general_2021}), we first review some basic concepts and definitions of $s$-level blocked designs. Let \( q = n - m \) and \( N = (s^q - 1)/(s - 1) \). A saturated design \( H_{qs} \) with \( q \) independent columns at level \( s \) can be formed by \( N \) points from the finite projective geometry \( PG(q-1,s) \), and its structure can be represented recursively as follows:
\begin{align}\label{H_{qs}}
H_{1s}=\{1\}, H_{rs}=\left\{H_{(r - 1)s}, r, (r, r^2, \ldots, r^{s - 1})H_{(r - 1)s}\right\}, \quad r= 2,\ldots, q. 
\end{align} 
We refer to (\ref{H_{qs}}) as the $r$th-order Yates order at level \( s \). Denote \(S_{qr} = H_{qs} \setminus H_{rs}\) and \(F_{qq} = H_q\setminus H_{q-1}\). In an experiment with \( p \) different sources of heterogeneity (i.e., \( p \) block variables), if the \( j \)th block variable divides the experimental units into \( s^{i_j} \) blocks, then \( i_j \) independent factors need to be assigned to represent the block effect. Thus, describing the entire block structure requires a total of \( \sum_{j=1}^p i_j \) factors, and these factors are not required to be mutually independent, which differs from single block variable designs. For simplicity, set \( i_j = 1 \) for \( j = 1, 2, \ldots, p \).

Let \( D = (D_t : D_b) \) denote an \( s^{n-m} : s^p \) blocked FF design with \( n \) treatment factors and \( p \) block factors at \( s \) levels, where \( s \) is a prime or prime power. Here, \( D_t \) denotes the unblocked \( s^{n-m} \) design determined by \( m \) independent treatment defining words, and \( D_b \) is the blocking scheme generated by \( p \) independent block defining words. A design \( D \) is said to have resolution \( R \) if no effect involving \( c \) factors is aliased with any other effect involving fewer than \( R - c \) factors. Let \( F_1, F_2, \dots, F_n \) be the \( n \) points in \( D_t \). For any \( 1 \leq i_1 < i_2 < \dots < i_j \leq n \), \( F_{i_1} \times F_{i_2} \times \cdots \times F_{i_j} \) is called a \( j \)th-order factor interaction ($j$fi), and each $j$fi can be further partitioned into \( (s-1)^{j-1} \) mutually orthogonal \( j \)th-order factor interaction components ($j$fic's). When \( j = 1 \), the treatment effect is called a treatment main effect. Similarly, the blocking scheme \( D_b \) also contains block main effects, and a \( j\)th-order block interaction ($j$bi) can be partitioned into \( j\)th-order block interaction components ($j$bic's). Any $j$fi in design \( D \) can be represented by an \( n \)-dimensional vector \( \mathbf{b} = (b_1, b_2, \ldots, b_n)^\top \) over the finite field \( GF(s) \). Two effects are considered identical if their corresponding vectors are proportional. Consequently, an \( s^{n-m}:s^p \) design contains a total of \( (s^n-1)/(s-1) \) treatment components and \( (s^p-1)/(s-1) \) block components. In designs with multi-block variables, it is typically assumed that interactions between treatment factors and block factors are negligible, and only block main effects and 2bic's are considered significant block components.

We extend the \(\text{B}^2\text{-ACNP}\) and \(\text{B}^2\text{-GMC}\) criterion to  \(s^{n-m}:s^p\) designs with multi-block variables, abbreviated as \(s^{n-m}:s^p\) MBV design. These criteria apply to any \(s\) level and require a resolution \(R \geq \text{III}\). Based on the aliasing relationships among component effects,  the \((s^n-1)/(s-1)\) treatment components are classified into four categories: 
 (i)    \(g\)-class: treatment components aliased with the overall mean \(I\), 
 (ii) \(b\)-class: treatment components aliased with block main effects or \(2\)bic's,  
 (iii) \(m\)-class: treatment components aliased with treatment main effects, and 
(iv) \(\phi\)-class: the remaining treatment components not covered above.
Within the four categories, a treatment component can be aliased only with other components in the same category. Under the orthogonal component system, any \(i\)fic can be aliased with other components at most \(K_j = \binom{n}{j}\) \(j\)fic's.
Let \(^{\#}{}_{i}^*C_j^{(k)}\) be the number of \(i\)fic's in the *-class, aliased with \(k\) \(j\)fic's, where \(*\in\{g,b,m,\phi\}\), \(i,j=0,1,\cdots,n\),  \(k=0,1,\cdots,K_j\). For each pair \((i,j)\), denote \(^{\#}{}_{i}^*C_j =( {^{\#}{}_{i}^*C_j^{(0)}}, {^{\#}{}_{i}^*C_j^{(1)}}, \cdots, {^{\#}{}_{i}^*C_j^{(K_j)}})\). For brevity, $h$ consecutive zeros in the elements of \({^{\#}{}_{i}^*C_j}\) are denoted as \(0^h\), and if all elements after a certain term are zero, the subsequent terms are omitted. The set
\begin{align}\label{B2ACNP}
{}_{\,}^{B}\!C=\{{^{\#}{}_{i}^gC_j},{^{\#}{}_{i}^bC_j},{^{\#}{}_{i}^mC_j},{^{\#}{}_{i}^\phi C_j},i,j=0,1,\cdots,n\}
\end{align}  
is called the blocked aliased component-number pattern (\(\text{B}^2\text{-ACNP}\)) of \(s^{n-m}:s^p\) design. It completely describes the aliasing information among various components in the blocked design \(D\).
\begin{remark}
The \(\text{B}^2\text{-ACNP}\) (\ref{B2ACNP}) differs from the \(\text{B}^1\text{-ACNP}\) in Li and Li (\citeyear{li_general_2025}), and the \(\text{B}^2\text{-AENP}\) in Zhang et al. (\citeyear{zhang_optimal_2011}). The \(\text{B}^2\text{-ACNP}\) is applicable to the MBV case, allows aliasing among block components, and only considers significant block components, i.e., block main effects and \(2\)bic's. The \(\text{B}^1\text{-ACNP}\) requires that block components are not aliased with each other and treats all block components as equally important, leading to differences in the \(b\)-class and \(\phi\)-class between the two. Furthermore, the \(\text{B}^2\text{-AENP}\) is only applicable to the two-level case.
\end{remark}

In what follows, we only consider block designs with resolution \(R \geq \text{III}\). Under this condition, there is no aliasing among the main effects in the \(m\)-class. If aliasing among block components is ignored in the study, then the aliasing situation in the \(b\)-class can be disregarded. Similarly, since the \(i\)fic's in the \(g\)-class can be expressed through the aliasing in the \(m\)-class (Zhang and Mukerjee \citeyear{zhang_general_2009}), they can also be neglected. Consequently, we only need to consider the aliasing structures in the \(m\)-class and the \(\phi\)-class. According to effect hierarchy principle, lower-order effects are more important than higher-order effects and effects of the same order are equally important, combined with the \(\text{B}^2\text{-ACNP}\), a \(s^{n-m}:s^p\) design is called a \(\text{B}^2\text{-GMC}\) design if it sequentially maximizes the components in the following sequence
\begin{align}\label{B2GMC}
{}_{\,}^{B}\!C(D)=({^{\#}{}_{1}^mC_2},{^{\#}{}_{2}^\phi C_2},{^{\#}{}_{1}^mC_3},{^{\#}{}_{2}^\phi C_3},{^{\#}{}_{3}^\phi C_2},{^{\#}{}_{3}^\phi C_3},\ldots).
\end{align}  

\begin{remark}
According to the classification of blocked designs in Bisgaard's (\citeyear{bisgaard_note_1994}), the case that ``block main effects and block interactions are equally significant, and any product of block letters is treated as a single letter" is referred to as $Kind~1$. Another case that satisfies the effect hierarchy principle for block effects, which states that ``lower-order block factors are more significant, the same order block factors are equally significant, and the \(\sum_{j=1}^p i_j\) factors are not required to be mutually independent" is referred to as $Kind~2$. Accordingly, it is clear that the B-GMC and \(\text{B}^1\text{-GMC}\) criteria, developed for single block variables, are discussed under the $Kind~1$ scenario, whereas the \(\text{B}^2\text{-GMC}\) criterion, intended for MBV, is addressed under the $Kind~2$.
\end{remark}

Next, we introduce notation to study the aliasing properties among components in MBV designs. Denote $\mathbf{S}=\{1,\cdots,s-1\}$. For a given design \(D\subset H_{qs}\) and \(\gamma \in H_{qs}\), define
$$
B_i(D,\gamma) = \# \left\{(d_1,\dots,d_i): d_1,d_2,\dots,d_i \in D, d_1(d_2)^{l_1}\cdots(d_i)^{l_{i-1}} = \gamma, l_1,\cdots,l_{i-1}\in \mathbf{S}\right\},
$$
which denotes the number of $i$fic's aliased with the component \(\gamma\) in the \(s\)-level design \(D\).  Here, \(\#\{A\}\) denotes the cardinality of set \(A\). In particular, for \(i = 2\), we have
\begin{align}\label{B_2(Q)}
B_2(D,\gamma) = \#\left\{(d_1,d_2,l) : d_1,d_2 \in D, \; d_1 d_2^l = \gamma, \; l \in \mathbf{S}\right\}.
\end{align} 
In experiments, the primary concern is usually the aliasing relationships among low-order components, such as main effects and $2$fic's. The aliasing analysis of third- and higher-order components is computationally more intensive; this section first employs equation (\ref{B_2(Q)}) to discuss the properties of aliasing among low-order components, while the analysis of higher-order components is handled by other methods.

Consider an MBV design \(D=(D_t:D_b)\) of the form \(s^{n-m}:s^p\), where \(D_t\) can be viewed as the unblocked \(s^{n-m}\) design consisting of selected columns from \(H_{qs}\). Let \(\overline{D}=H_{qs}\setminus D = H_{qs}\setminus (D_t \cup D_b)\). From Li et al. (2015), it follows that
\begin{equation}
\begin{aligned}\label{ACNP}
\begin{cases} 
{}_{\,1}^{\#}\!C_2^{(k)}(D_t)=\#\{\gamma:\gamma\in D_t,B_2(D_t,\gamma)=k \},\\
{}_{\,2}^{\#}\!C_2^{(k)}(D_t)=(k + 1)\#\{\gamma:\gamma\in H_{qs},B_2(D_t,\gamma)=k+1 \}.
\end{cases}
\end{aligned}
\end{equation}
Denote the set \(U = D_b \cup \overline{D}\), and \(U(D_b) = \left\{ \gamma \in H_{qs}: \gamma \in D_b \text{ or } \gamma = d_1 d_2^l~\text{with}~d_1, d_2 \in D_b, l \in  \mathbf{S} \right\}\). It is easy to see that the cardinalities of these two sets are \(\#\{U\} = (s^q - 1)/(s - 1) - n\) and \(\#\{U(D_b)\} \leq p + p(p-1)(s-1)/2\), where the equality holds for the latter if and only if the block components are mutually independent. 
Then, by the definition of the \(\text{B}^2\text{-ACNP}\), it yields that
\begin{equation}\label{b2two}
\begin{aligned}
\begin{cases} 
^{\#}{}_{1}^mC_2^{(k)}(D)=\#\{\gamma:\gamma\in D_t,B_2(D_t,\gamma)=k \},\\
^{\#}{}_{2}^b C_2^{(k)}(D)=(k + 1)\#\{\gamma:\gamma\in U(D_b),B_2(D_t,\gamma)=k+1 \},\\
^{\#}{}_{2}^\phi C_2^{(k)}(D)=(k + 1)\#\{\gamma:\gamma\in U\setminus U(D_b),B_2(D_t,\gamma)=k+1 \}.
\end{cases}
\end{aligned}
\end{equation}
To discuss the construction problem of \(\text{B}^2\text{-GMC}\) designs, we analyze the relationship between the \(\text{B}^2\text{-ACNP}\) of design \(D\) and the ACNP of design \(D_t\).

\begin{theorem}\label{th6}
	For any \(s^{n-m}:s^p\) MBV design \(D=(D_t:D_b)\), we have
		  $${^{\#}{}_{1}^mC_2^{(k)}(D)}={{}_{1}^{\#}C_2^{(k)}(D_t)}, \quad {^{\#}{}_{2}^\phi C_2^{(k)}(D)}={{}_{2}^{\#}C_2^{(k)}(D_t)}-(k+1){{}_{1}^{\#}C_2^{(k+1)}(D_t)-{^{\#}{}_{2}^bC_2^{(k)}(D)} }.$$
\end{theorem}
\begin{proof}
	Based on the equations (\ref{ACNP}) and (\ref{b2two}), it is easy to get \({^{\#}{}_{1}^mC_2^{(k)}(D)}={{}_{1}^{\#}C_2^{(k)}(D_t)}\).  
	Since \(H_{qs} = D_t \cup U(D_b) \cup (U\backslash U(D_b))\), combining equations (\ref{ACNP}) and (\ref{b2two})  yields
	\begin{equation}
	\begin{aligned}
	{}_{\,2}^{\#}\!C_2^{(k)}(D_t)=&(k + 1)\#\{\gamma:\gamma\in D_t \cup U(D_b) \cup U\backslash U(D_b),B_2(D_t,\gamma)=k+1 \}\\
	=&(k + 1)\#\{\gamma:\gamma\in D_t ,B_2(D_t,\gamma)=k+1 \} + (k + 1)\#\{\gamma:\gamma\in U(D_b) ,B_2(D_t,\gamma)=k+1 \} \\
	&+ (k + 1)\#\{\gamma:\gamma\in U\backslash U(D_b),B_2(D_t,\gamma)=k+1 \}\\
	=&(k+1){{}_{1}^{\#}C_2^{(k+1)}(D_t)}+{^{\#}{}_{2}^bC_2^{(k)}(D)}+{^{\#}{}_{2}^\phi C_2^{(k)}(D)}.\nonumber
	\end{aligned}
	\end{equation}
	Rearranging the terms proves the second equation.\hfill $\qed$
\end{proof}

Through Theorem \ref{th6} and equation (\ref{B2GMC}), an \(s^{n-m}:s^p\) MBV design \(D=(D_t:D_b)\) is a \(\text{B}^2\text{-GMC}\) design only if its corresponding non-blocked \(s^{n-m}\) design \(D_t\) satisfies sequentially maximizing \({{}_{1}^{\#}C_2(D_t)}\). After determining \({{}_{1}^{\#}C_2(D_t)}\), sequentially maximizing \({^{\#}{}_{2}^\phi C_2(D)}\) is equivalent to sequentially maximizing \({{}_{2}^{\#}C_2(D_t)}-{^{\#}{}_{2}^bC_2(D)} \).  Thus, we directly get the following result.

\begin{theorem}\label{th7}
	When only treatment main effects and 2fic's are considered, an \(s^{n-m}:s^p\) MBV design \(D=(D_t:D_b)\) is a \(\text{B}^2\text{-GMC}\) design if and only if it sequentially maximizes the sequence \(({{}_{1}^{\#}C_2(D_t)}, {{}_{2}^{\#}C_2(D_t)}-{^{\#}{}_{2}^bC_2(D)} )\).
\end{theorem}

When the number of factors \(n\) is moderate, the aliasing information among low-order components can be directly obtained from formulas (\ref{ACNP}) and (\ref{b2two}). Combined with Theorem 2 and a search procedure, a B\(^2\)-GMC design can be identified. However, when \(n\) is large, the computational complexity of the B\(^2\)-ACNP increases significantly. By exploiting the relationship between \(D_t\) and its complement \(U = H_{qs} \backslash D_t\), we propose an efficient method for computing low-order component aliasing in this setting.

\begin{lemma}\label{lem1}
	For an $s^{n-m}:s^p$ MBV design $D = (D_t, D_b)$, we have
	\[
	B_2(D_t,\gamma) = 
	\begin{cases}
	B_2(U,\gamma) + (s-1)(n-1) - (s^{n-m}-s)/2, & \gamma \in D_t,\\
	B_2(U,\gamma) + (s-1)n - (s^{n-m}-s)/2, & \gamma \in U.
	\end{cases}
	\]
\end{lemma}

\begin{proof}
	By the definition of \(B_2(Q, \gamma)\), we have
	\[
	B_2(H_{(qs)},\gamma) = B_2(D_t,\gamma) + B_2(U,\gamma) + \#\bigl\{(d_1,d_2,l) : d_1 \in D_t, d_2 \in U, d_1 d_2^l = \gamma, l \in \mathbf{S}\bigr\}.
	\]
	When \(\gamma \in D_t\), we have \(\#\{D_t\} = n\). In the cross terms, the element \(d_1 \in D_t\) can be chosen from \(n-1\) possibilities, and \(l \in \mathbf{S}\) provides \((s-1)\) choices. Thus, the total number of possible choices for \(d_2 \in H_{qs}\) is \((s-1)(n-1)\), among which \(B_2(D_t,\gamma)\) pairs correspond to \(d_2 \in D_t\). Consequently, the number of cross terms is \((s-1)(n-1) - 2B_2(D_t,\gamma)\). Combining this with Lemma 1 of Li et al. (\citeyear{li_analysis_2016}), the original expression becomes
	\[
	B_2(D_t,\gamma) = 
	B_2(U,\gamma) + (s-1)(n-1) - (s^{n-m}-s)/2, \quad \gamma \in D_t.
	\]
	When \(\gamma \in U\), in the cross terms, the element \(d_1 \in D_t\) can be chosen from \(n\) possibilities, leading to \((s-1)n\) possible choices for \(d_2 \in H_{qs}\). Hence, the number of cross terms is \((s-1)n - 2B_2(D_t,\gamma)\), and we obtain
	\[
	B_2(D_t,\gamma) = 
	B_2(U,\gamma) + (s-1)n - (s^{n-m}-s)/2, \quad \gamma \in U.
	\]
	This completes the proof.\hfill $\qed$
\end{proof}

Based on Lemma \ref{lem1}, we can express \(B_2(D_t, \gamma)\) for an \(s\)-level design \(D\) in terms of \(B_2(U, \gamma)\), thereby computing \({^{\#}{}_{1}^m C_2^{(k)}}\) and \({^{\#}{}_{2}^\phi C_2^{(k)}}\) via a smaller \(\#\{U\}\) when \(\#\{D_t\}=n\) is large.

\begin{theorem}\label{th8}
	For an $s^{n-m}:s^p$ MBV design $D = (D_t, D_b)$, we have
	\begin{equation}\label{b2twoU}
	\begin{aligned}
	\begin{cases} 
	{^{\#}{}_{1}^m C_2^{(k)}(D)} = \#\bigl\{ \gamma : \gamma \in D_t, B_2(U,\gamma) = k - (s-1)(n-1) + (s^{n-m}-s)/2 \bigr\},\\
	{^{\#}{}_{2}^\phi C_2^{(k)}(D)} = (k+1)\#\bigl\{ \gamma : \gamma \in U \backslash U(D_b), B_2(U,\gamma) = k+1 - (s-1)n + (s^{n-m}-s)/2 \bigr\}.
	\end{cases}
	\end{aligned}
	\end{equation}
\end{theorem}

Obviously, when \(n\) is large, it is more convenient to calculate \({^{\#}{}_{1}^m C_2^{(k)}}\) and \({^{\#}{}_{2}^m C_2^{(k)}}\) using Theorem \ref{th8}. We illustrate the application of Theorem \ref{th8} with an example.

\begin{example}\label{exm1}
	Consider a \(3^{31-27}:3^2\) design \(D = (D_t, D_b)\), with
	$D_t=F_{44}\cup \{3,13,23,123 \}$, $D_b=H_{23}$,  and $U=H_{33}\backslash  \{3,13,23,123 \}$.
	Clearly, any factor in the set \(F_{44}\) of \(D_t\) is not aliased with any 2fic in \(U\). Each factor in the set \(\left\{3,13,23,123\right\}\) is aliased with six 2fic's in \(U\). Based on Theorem \ref{th8}, we have
	\[
	^{\#}{}_{1}^mC_2^{(k)}(D) =
	\begin{cases}
	27, & k=21, \\
	4, & k=27, \\
	0, & \text{otherwise},
	\end{cases}
	\quad
	^{\#}{}_{2}^\phi C_2^{(k)}(D) =
	\begin{cases}
	112, & k=27, \\
	29, & k=28, \\
	0, & \text{otherwise}.
	\end{cases}
	\]
\end{example}

Through Theorem \ref{th8}, for some specific designs, we can directly obtain the aliasing structure among their low-order components, which is helpful for the study of constructing \(\text{B}^2\text{-GMC}\) designs at \(s\) levels.

\begin{theorem}
	Let \(D = (D_t : D_b)\) be an \(s^{n-m}: s^p\) MBV design with \(q = n-m\). Suppose there exists an integer \(v\) with \(1 \leq v \leq q-1\) such that \(D_t = S_{qv}\) and \(D_b \subset H_v\).  Then,  
	$$  {^{\#}{}_{1}^m C_2^{(k)}(D)} = 
	\begin{cases}
	\displaystyle \frac{s^q - s^v}{s - 1}, & k = K_1, \\
	0, & \text{otherwise},\quad 
	\end{cases}
 {^{\#}{}_{2}^\phi C_2^{(k)}(D)} = 
	\begin{cases}
	\displaystyle \frac{s^q - s^v}{2} \left( \frac{s^v - 1}{s - 1} - M \right), & k = K_2 - 1, \\
	0, & \text{otherwise},
	\end{cases}
		$$
	where
	\(
	K_1 = (s^q - 2s^v - s + 2)/2,~K_2 = (s^q - s^v)/2,
	\) and \(M=\#\{U(D_b)\}\).
\end{theorem}
\begin{proof}
	Because \(H_v\) is a \(v\)-flat of \(H_q\), it is closed under the formation of two‑factor interactions. Combined with Lemma 1 of Li et al. (\citeyear{li_analysis_2016}), for any \(\gamma \in H_q\), we have
	\begin{equation}\label{th4.1}
	    B_2(H_v, \gamma) = 
	    \begin{cases}
	    \displaystyle (s^v - s)/2, & \gamma \in H_v, \\
	    0, & \gamma \notin H_v.
	    \end{cases}
	\end{equation}
	Note that \(U = H_q \backslash D_t = H_v\). Applying Lemma \ref{lem1} to the present case, we obtain for \(\gamma \in D_t\), $B_2(D_t, \gamma) = B_2(H_v, \gamma) + (s-1)(n-1) - (s^q - s)/2$.
	Since $\#\{U\}=(s^v-1)/(s-1)$, we have $n =(s^q-1)/(s-1)-(s^v-1)/(s-1)= (s^q - s^v)(s-1)$. Using Lemma \ref{lem1} with \(\gamma \notin H_v\) and substituting \(n = (s^q - s^v)(s-1)\), a straightforward algebraic simplification yields
	\begin{equation}\label{th4.2}
		B_2(D_t, \gamma) = (s^q - 2s^v - s + 2)/2 = K_1 \quad \text{for all } \gamma \in D_t.
	\end{equation}
	Similarly, for \(\gamma \in H_v\), Lemma \ref{lem1} gives
	$B_2(D_t, \gamma) = B_2(H_v, \gamma) + (s-1)n - (s^q - s)/2$. Inserting the non‑zero branch of (\ref{th4.1}) and the expression for \(n\), we obtain
	\begin{equation}\label{th4.3}
		B_2(D_t, \gamma) = \frac{s^q - s^v}{2} = K_2 \quad \text{for all } \gamma \in H_v.
	\end{equation}
	Since for all \(\gamma \in D_t\), \(B_2(D_t, \gamma)\) is constantly equal to the constant \(K_1\), this set is nonempty only when \(k = K_1\), and the size of the set is exactly equal to the number of elements in \(D_t\), i.e., \(n\). Then, combining with equation (\ref{b2twoU}), we obtain the expression of 
	${^{\#}{}_{1}^m C_2^{(k)}(D)}$.
	
	Note that \(U = H_v\) and \(U(D_b) \subset H_v\) by condition (\ref{th4.2}). The set \(U \backslash U(D_b)\) then has cardinality \((s^v - 1)/(s-1) - M\). For every \(\gamma\) in this set, (\ref{th4.3}) gives \(B_2(D_t, \gamma) = K_2\). Thus, the expression of 
	${^{\#}{}_{2}^\phi C_2^{(k)}(D)}$ follows.
	 \hfill $\qed$
\end{proof}

\begin{corollary}\label{co1}
	Let \(D = (D_t : D_b)\) be an \(s^q: s^p\) MBV design with \(D_t =S_{q(q-1)}=F_{qq}\) and \(D_b \subset H_{q-1}\). Then,
	\begin{align*}
	   {^{\#}{}_{1}^m C_2^{(k)}(D)} &= 
	\begin{cases}
	\displaystyle s^{q-1}, & k = {(s-2)(s^{q-1}-1)}/2, \\
	0, & \text{otherwise},
	\end{cases} \\
	{^{\#}{}_{2}^\phi C_2^{(k)}(D)} &= 
	\begin{cases}
	\displaystyle \frac{s^{q-1}(s-1)}{2} \left( \frac{s^v - 1}{s - 1} - M \right), & k = s^{q-1}(s-1)/2-1, \\
	0, & \text{otherwise}.
	\end{cases} 
	\end{align*}
\end{corollary}

\begin{example}\label{exm2}
	(i) Consider a \(2^{16-11}: 2^3\) design \(D = (D_t: D_b)\) with \(D_t = {F_{55}}\) and \(D_b = H_{22} \subset H_{42}\). According to Corollary \ref{co1}, we have \(U(D_b)=H_{22}\), \(M=3\), and
	\[
	^{\#}{}_{1}^mC_2^{(k)}(D) =
	\begin{cases}
	16, & k=0, \\
	0, & \text{otherwise},
	\end{cases}
	\quad
	^{\#}{}_{2}^\phi C_2^{(k)}(D) =
	\begin{cases}
	96, & k=7, \\
	0, & \text{otherwise}.
	\end{cases}
	\]
	
	(ii) Consider a \(3^{27-23}: 3^3\) design \(D = (D_t: D_b)\) with \(D_t = {F_{44}}\) and \(D_b = \{1,2,3\} \subset H_{42}\). Following Corollary \ref{co1}, we have \(U(D_b)=\{1, 2, 3, 12, 12^2,13,13^2,23,23^2\}\), \(M=9\), and
	\[
	^{\#}{}_{1}^mC_2^{(k)}(D) =
	\begin{cases}
	27, & k=13, \\
	0, & \text{otherwise},
	\end{cases}
	\quad
	^{\#}{}_{2}^\phi C_2^{(k)}(D) =
	\begin{cases}
	108, & k=26, \\
	0, & \text{otherwise}.
	\end{cases}
	\]
\end{example}

The above analysis is efficient for investigating confounding among main effects and 2fic's.  As the order of components increases, the computational complexity of aliasing among components increases significantly (Huang et al., \citeyear{huang_lower-order_2023}; Yan et al., \citeyear{YAN2026117262}). To address this, we propose a novel method to analyze the aliasing structure among higher-order components.

To systematically compute the aliasing relationships among high-order components, we introduce the blocked wordlength distribution matrix (B-WDM). For an MBV design \(D = (D_t: D_b)\), each alias set contains a component of \(H_{qs}\). Let  \(l\) be an index of the alias set that contains the \(l\)th component of \(H_{qs}\), and denote \(T_l = (t_1, t_2, \dots, t_{s^m})^\top\) as the wordlengths of all words in that alias set, where the wordlength of the component itself is placed at \(t_1\). The vector \(T_0 = (t_1, t_2, \dots, t_{(s^m-1)/(s-1)})^\top\) corresponds to the wordlengths of the components in the treatment defining contrast subgroup. All \(T_l\) are classified into four categories: (i) the \(g\)-class, consisting solely of \(T_0\); (ii) the \(b\)-class, containing components aliased with significant block effects; (iii) the \(m\)-class, containing components aliased with main effects, i.e., \(T_l\) that contain the element 1; (iv) the \(\phi\)-class, consisting of all remaining \(T_l\). The collection \(T = (T_1, \ldots, T_N)\) is called the B-WDM.

For \(* \in \{g, b, m, \phi\}\), define \(F_i(T_{*l}, j)\) as the number of \(j\)-th order component effects contained in the alias set corresponding to \(T_{*l}\), provided that the alias set contains at least one \(i\)-th order component effect; otherwise, \(F_i(T_{*l}, j)=0\). Then, for any MBV \(s^{n-m}:s^p\) design, we have

\begin{align}\label{T_l&icjk}
^{\#}{}_{i}^*C_j^{(k)}=\sum_{*l}\boldsymbol{1}_{\{F_i(T_{*l}, j)=k+\delta_{ij}\}} F_i(T_{*l}, i),
\end{align}
where \(\mathbf{1}_{\{\cdot\}}\) is the indicator function; \(\delta_{ij}=1\) if \(i=j\), and 0 otherwise. In particular, for the aliasing of low-order components, it follows that
\begin{align}
^{\#}{}_{1}^*C_j^{(k)}&=\sum_{*l}\boldsymbol{1}_{\{F_1(T_{*l}, j)=k\}}F_1(T_{*l}, 1) =\#\{T_{*l}: F_1(T_{*l}, j)=k, F_1(T_{*l}, 1)\neq 0\},~j\geq 2,\label{C1jk}\\
^{\#}{}_{i}^*C_i^{(k)}&=(k + 1)\sum_{*l}\boldsymbol{1}_{\{F_i(T_{*l}, i)=k+1\}}=(k + 1)\#\left\{T_{*l}: F_i(T_{*l}, i)=k + 1\right\}.\label{C22k}
\end{align}

\begin{example}\label{exm3}
	(i) Consider a \(2^{5-1}:2^{3}\) MBV design \(D=(D_{t},D_{b})\) with 
	\(D_{t}=\{12345\}\) and \(D_{b}=\{12(=b_{1}),134(=b_{2}),234(=b_{3})\}\). 
	Its defining contrast subgroup is 
	\(G=\{I,12345,12b_{1},134b_{2},234b_{3},234b_{1}b_{2},134b_{1}b_{3},12b_{2}b_{3},b_{1}b_{2}\\b_{3},345b_{1},25b_{2},15b_{3},15b_{1}b_{2},25b_{1}b_{3},345b_{2}b_{3},12345b_{1}b_{2}b_{3}\}\).
	Based on the B-WDM, the alias sets \(T_{l}\) are classified into four classes as follows: 
	(i) \(g\)-class: \(T_0=(0,5)\);
	(ii) \(b\)-class: \(T_3=(2,3), T_{13}=(3,2), T_{14}=(3,2)\);
	(iii) \(m\)-class: \(T_1=(1,4), T_2=(1,4), T_4=(1,4), T_8=(1,4), T_{15}=(4,1)\);
	(iv) \(\phi\)-class: \(T_5=(2,3), T_6=(2,3), T_7=(3,2), T_9=(2,3), T_{10}=(2,3), T_{11}=(3,2), T_{12}=(2,3)\). 
	By equation \eqref{C1jk}, \(^{\#}{}_{1}^{m}C_{2}^{(0)}=\#\{T_{1},T_{2},T_{4},T_{8},T_{15}\}=5\), and for \(k\geq 2\), \(^{\#}{}_{1}^{m}C_{2}^{(k)}=0\); hence \(^{\#}{}_{1}^{m}C_{2}=(5)\). Similarly, \(^{\#}{}_{1}^{m}C_{3}=(5)\), \(^{\#}{}_{1}^{m}C_{4}=(0,5)\), \(^{\#}{}_{1}^{m}C_{5}=(5)\). By equation \eqref{C22k}, \(^{\#}{}_{2}^{\phi}C_{2}^{(0)}=\#\{T_5,T_6,T_7,T_9,T_{10},T_{11},T_{12}\}=7\) and \(^{\#}{}_{2}^{\phi}C_{2}^{(k)}=0\) for \(k\geq 2\); thus \(^{\#}{}_{2}^{\phi}C_{2}=(7)\).
	
	(ii) Consider a \(3^{5-2}:3^{2}\) MBV design \(D=(D_{t},D_{b})\) with 
	\(D_{t}=\{1234^{2},12^{2}5^{2}\}\) and \(D_{b}=\{13(=b_{1}),12(=b_{2})\}\). 
	Equations \eqref{C1jk} and \eqref{C22k} yield
	\( {^{\#}{}_{1}^m C_2}=(2,3) \),
	\( {^{\#}{}_{1}^m C_3}=(0^3,5) \),
	\( {^{\#}{}_{1}^m C_4}=(0^2,3,0^2,2) \),
	\( {^{\#}{}_{1}^m C_5}=(2,0,3) \), and
	\( {^{\#}{}_{2}^\phi C_2}=(1, 6) \).
\end{example}

\section{Relaionship with other criteria}\label{sec3}

Since the \(\text{B}^2\text{-ACNP}\) can characterize the aliasing structure more precisely, this section systematically analyzes the relationships between the \(\text{B}^2\text{-ACNP}\) and the classical block criteria MA, CE, and MEC in the context of MBV designs of type \(Kind~2\), based on the \(\text{B}^2\text{-ACNP}\) itself and the function \(\{F_i(T_{*l},j), i,j=0,1,\ldots,n, *=g,b,m,\phi\}\). Furthermore, expressions for these classification patterns are provided. 
According to Equation (\ref{T_l&icjk}), we provide a   convenient method to calculate  \({^{\#}{}_{i}^*C_0^{(k)}}\) and \({^{\#}{}_{0}^*C_j^{(k)}}\)(\(*=g,b,m,\phi\)).

\begin{theorem}\label{th1}
	For any \(s^{n-m}:s^p\) MBV design  \(D=(D_t:D_b)\), we have
	\begin{flalign}
	(i)\quad{^{\#}{}_{0}^gC_j}&=(0^{F_0(T_0,j)-\delta_{0j}},1),~{^{\#}{}_{0}^*C_j}=(0),~j=0,1,\cdots n,*=b,m,\phi.\nonumber&\\
	(ii)\quad{^{\#}{}_{i}^gC_0}&=(0,F_i(T_0,i)),~{^{\#}{}_{i}^*C_0}=(\sum_{*l}F_i(T_{*l},i)),~i=1,\cdots n,*=b,m,\phi.\nonumber&
	\end{flalign}
\end{theorem}

\begin{proof}
	Since \({^{\#}{}_{0}^gC_j^{(k)}}\) represents the number of identity elements in the $g$-class that is aliased with \(k\) $j$fic's, according to Equation (\ref{T_l&icjk}), for \(j=0,1,\cdots,n\), we have
	\[
	{^{\#}{}_{0}^gC_j^{(k)}}=\sum_{gl} \mathbf{1}_{\{F_{0}(T_{gl}, j)=k+\delta_{0j}\}}  F_{0}(T_{gl}, 0)=\mathbf{1}_{\{F_{0}(T_{0}, j)=k+\delta_{0j}\}}= \begin{cases}1, k=F_{0}(T_{0}, j)-\delta_{0 j},\\
	0, k \neq F_{0}\left(T_{0}, j\right)-\delta_{0 j}.\end{cases}
	\]
In the $*$-class, there is no identity element at all, thus it is obvious that \({^{\#}{}_{0}^*C_j}=(0)\). This proves (i). 

	When \(k=0,1\) and \(j=1,2,\cdots,n\), from Equation (\ref{T_l&icjk}), we obtain
	\begin{align}
	{^{\#}{}_{i}^gC_0^{(0)}}&=\sum_{gl} \mathbf{1}_{\{F_{i}(T_{gl}, 0)=0+\delta_{i0}\}}  F_{i}(T_{gl}, i)=\mathbf{1}_{\{F_{i}(T_{0}, 0)=0\}}  F_{i}(T_{0}, i)=0,\nonumber\\
	{^{\#}{}_{i}^gC_0^{(1)}}&=\sum_{gl} \mathbf{1}_{\{F_{i}(T_{gl}, 0)=1+\delta_{i0}\}}  F_{i}(T_{gl}, i)=\mathbf{1}_{\{F_{i}(T_{0}, 0)=1\}}  F_{i}(T_{0}, i)=F_{i}(T_{0}, i).\nonumber
	\end{align}
	Since there is only one identity element, when \(k\geq2\), we have \({^{\#}{}_{i}^gC_0^{(k)}}=0\). Moreover, since all $i$fic's in the $*$-class are not aliased with the identity element, and the total number of $i$fic's in the $*$-class is \(\sum_{*l}F_i(T_{*l},i)\), it follows that \({^{\#}{}_{i}^*C_0}=(\sum_{*l}F_i(T_{*l},i))\). This proves (ii).\hfill $\qed$
\end{proof}

Let \(A_i\) and \(B_i\) denote the numbers of words of length \(i\) in the treatment-defining words and the block defining contrast subgroup, respectively. Based on these, four blocked MA criteria from different perspectives can be referred to Sitter et al. (\citeyear{sitter_fractional_1997}), Chen and Cheng (\citeyear{chen1999theory}), Zhang and Park (\citeyear{zhang2000optimal}), and Cheng and Wu (\citeyear{cheng2002choice}). These criteria are achieved by sequentially minimizing the following sequences:
\begin{equation}
\begin{aligned}\label{MA}
W_{scf} &= (A_3, B_2, A_4, ..., A_n, B_{n-1}),  \\
W_{cc} &= (3A_3 + B_2, A_4, 10A_5 + B_3, A_6, ...), \\
W_{zp}^{cw} &= (A_3, B_2, A_4, A_5, B_3, A_6, ..., A_{2j-1}, B_j, A_{2j}, ...), \\
W_{cw} &= (A_3, A_4, B_2, A_5, A_6, B_3 ..., A_{2j-1}, A_{2j}, B_j, ...).
\end{aligned}
\end{equation}
From equation (\ref{MA}), it can be seen that \(A_i\) and \(B_i\) play important roles in all four MA criteria. Wei et al. (\citeyear{wei_blocked_2014}) obtained the relationships \(A_i = {^{\#}{}_{i}^gC_0^{(1)}}\) and \(B_i = \sum_{k=0}^{K_{i-1}}{^{\#}{}_{i}^bC_i^{(k)}}\) for two-level designs with a single block variable. Next, we extend these to the case of \(s\)-level designs with MBVs.

\begin{theorem}\label{th2}
	Let \( D=(D_t:D_b) \) be an \( s^{n-m}:s^p \) MBV design. Then,
	$$ A_i={^{\#}{}_{i}^gC_0^{(1)}}=F_i(T_0,i)=F_0(T_0,i),\quad
B_i=\sum_{k=0}^{K_{i}-1}{^{\#}{}_{i}^bC_i^{(k)}}={^{\#}{}_{i}^bC_0^{(0)}}=\sum_{bl}F_i(T_{bl},i).$$
\end{theorem}

\begin{proof}
	From the definitions of \(A_i\) and the \(\text{B}^2\text{-ACNP}\), \(A_i\) is precisely the number of \(i\)fic's in the \(g\)-class. Thus, \(A_i = {^{\#}{}_{i}^gC_0^{(1)}}\). It further follows from Theorem \ref{th1} and the definition of \(F_i(T_{*l},j)\). By definition, \(B_i\) is the total number of \(i\)fic's in the \(b\)-class. Under the orthogonal component system, an \(i\)fic can be aliased with at most \(K_j-1\) other \(i\)fic's. Hence, \(B_i = \sum_{k=0}^{K_{i}-1}{^{\#}{}_{i}^bC_i^{(k)}}\). Since there is no identity element in the \(b\)-class, \({^{\#}{}_{i}^bC_0^{(0)}}\) also represents the total number of \(i\)fic's in the \(b\)-class. The last equation follows from Theorem \ref{th1}.\hfill $\qed$
\end{proof}

To establish connections between various criteria, the subsequent theorems in this section focus exclusively on  main effects and 2fic's, ignoring third- and higher-order $i$fic's. A low-order treatment component is said to be clear if it is not aliased with any significant block component or the identity element, nor with any other low-order treatment component. Let \(C_1\) and \(C_2\) denote the numbers of clear main effects and clear 2fis, respectively, and  \(CC\) denote the number of clear 2fic's. The following result can then be obtained.

\begin{theorem}\label{th3}
	For any \(s^{n-m}:s^p\) MBV design \(D=(D_t:D_b)\) of resolution \(R \geq III\), we have
	\begin{flalign}
	(i)&\quad C_1(D)={^{\#}{}_{1}^mC_2^{(0)}}(D)=\#\{ T_{ml} : F_1(T_{ml}, 1) = 1 \text{ and } F_1(T_{ml}, 2) = 0 \}.\nonumber&\\
	(ii)&\quad CC(D)={^{\#}{}_{2}^\phi C_2^{(0)}}(D)=\#\{ T_{\phi l} : F_2(T_{\phi l}, 2) = 1\} .\nonumber&
	\end{flalign}
\end{theorem}
\begin{proof}
	In a design with resolution \(R \geq III\), main effects are not aliased with each other. Thus, the number of main effects in the \(m\)-class that is not aliased with any 2fic's equals the number of clear main effects. Hence, \(C_1(D) = {^{\#}{}_{1}^mC_2^{(0)}}(D)\). From the definition of \(CC(D)\), clear 2fic's can only appear in the \(\phi\)-class. Thus, the number of 2fic's in the \(\phi\)-class that are not aliased with other 2fic's, \({^{\#}{}_{2}^\phi C_2^{(0)}}(D)\), is exactly \(CC(D)\). Combining equations (\ref{C1jk}) and (\ref{C22k}), the result is proved.\hfill $\qed$
\end{proof}
 
\begin{corollary}\label{cor1}
	For any \(2^{n-m}:2^p\) MBV design \(D=(D_t:D_b)\) of resolution \(R \geq III\), we have
	\begin{flalign}
	(i)\quad C_1(D)&={^{\#}{}_{1}^mC_2^{(0)}}(D)=\#\{ T_{ml} : F_1(T_{ml}, 1) = 1 \text{ and } F_1(T_{ml}, 2) = 0 \} ,\nonumber&\\
	(ii)\quad C_2(D)&={^{\#}{}_{2}^\phi C_2^{(0)}}(D)=\#\{ T_{\phi l} : F_2(T_{\phi l}, 2) = 1\} .\nonumber&
	\end{flalign}
\end{corollary}

\begin{remark}\label{re3}
	From equation (\ref{B2GMC}), it is easy to see that for designs with resolution \(R \geq III\), the \(\text{B}^2\text{-GMC}\) design must contain the maximum number of \(C_1\). When \(R \geq IV\), the \(\text{B}^2\text{-GMC}\) design has \(C_1 = {^{\#}{}_{1}^{m}C_2^{(0)}} = n\). Since the \(\text{B}^2\text{-GMC}\) design sequentially maximizes the sequence (\ref{B2ACNP}), it must simultaneously contain the maximum numbers of both \(C_1\) and \(CC\).
\end{remark}

\begin{theorem}
	Let \(D = (D_t : D_b)\) be an \(s^{n-m}: s^p\) MBV design with \(q = n-m\). If \(n \geq (s^{q-1} - 1)/(s - 1) + 2\), then the resolution \(R\) of design \(D\) satisfies \(R \leq III\), and \(C_1 = C_2 = CC = 0\).
\end{theorem}

\begin{proof}
	From Lemma 1 of Li et al. (2016), the columns in \(H_q \setminus \{\gamma\}\) can be partitioned into \((s^{q-1}-1)/(s-1)\) disjoint sets such that any two columns from the same set form a 2fic that is aliased with \(\gamma\). Since \(\#\{D_t\} = n \geq (s^{q-1}-1)/(s-1)+2\), placing the \(n\) treatment columns into these \((s^{q-1}-1)/(s-1)\) sets forces at least one set to contain two distinct treatment columns, say \(d_1\) and \(d_2\). Consequently, the 2fic formed by \(d_1\) and \(d_2\) is aliased with \(\gamma\).  
	If we take \(\gamma\) to be a main effect, then it is aliased with at least one 2fic, so \(C_1 = 0\). If we instead take \(\gamma\) to be a 2fic, then it is aliased with another 2fic. In summary, \(R \leq III\) and \(CC = 0\).\hfill $\qed$
\end{proof}

Similar to the comparison method between the two-level \(\text{B}^1\text{-GMC}\) and \(\text{B}^2\text{-GMC}\) criteria in Zhang et al. (\citeyear{zhang_optimal_2011}), and combined with Theorem \ref{th3}, we propose the following theorem.

\begin{theorem}\label{th4}
	Suppose that third- and higher-order treatment components as well as block components are negligible in the $Kind~2$ situation for an \(s^{n-m}:s^p\) design. The following results hold.
    \\
	(i) For \(p \leq 2\), the \(\text{B}^1\text{-GMC}\) design and the \(\text{B}^2\text{-GMC}\) design are identical for any \(n\), \(m\), and \(s\). \\
	(ii) For given \(n\), \(m\), \(s\), \(p > 2\) and \(R \geq III\), the \(\text{B}^2\text{-GMC}\) design must have a \(C_1\) greater than or equal to that of the \(\text{B}^1\text{-GMC}\) design. When \(R \geq IV\), both \(C_1\) and \(CC\) of the \(\text{B}^2\text{-GMC}\) design are greater than or equal to the corresponding numbers of the \(\text{B}^1\text{-GMC}\) design.
\end{theorem}

\begin{proof}
	By the definitions of $Kind~1$, $Kind~2$, and the \(\text{B}^1\text{-GMC}\) and \(\text{B}^2\text{-GMC}\) criteria, when \(p \leq 2\), there are no third- or higher-order block interaction components. Thus, all block components are significant. Consequently, for \(p \leq 2\), the candidate designs for $Kind~1$ and $Kind~2$ are identical, and the corresponding \(\text{B}^1\text{-GMC}\) and \(\text{B}^2\text{-GMC}\) criteria are also the same. Thus, part (i) is proved.  
    
	Let \(\mathscr{S}_i(n,m,s)(i=1,2)\) be the sets of candidate  \(s^{n-m}:s^p\) designs under the $Kind~1$ or 2 situation. Since $Kind~2$ allows aliasing relationships among block factors, we have \(\mathscr{S}_1(n,m,s) \subseteq \mathscr{S}_2(n,m,s)\). For a given design, according to the definitions of the \(*_i\)-classes with \(*=g,b,m\) and \(\phi\), and since $Kind~2$ considers only significant block components, for \(p > 2\) we have the inclusion relationship
	\begin{equation}
	\begin{aligned}\label{kindclass}
	g_1\text{-class} \supseteq g_2\text{-class}, \quad
	b_1\text{-class} \supseteq b_2\text{-class}, \quad
	m_1\text{-class} \subseteq m_2\text{-class}, \quad
	\phi_1\text{-class} \subseteq \phi_2\text{-class}.
	\end{aligned}
	\end{equation}
	Furthermore, for a given design, changing from $Kind~1$ to $Kind~2$ (or vice versa) alters only the component classification and does not change the specific aliasing relationships within an alias set. Let \(D^*_1\) be the \(\text{B}^1\text{-GMC}\) design selected from the set \(\mathscr{S}_1(n,m,s)\). Since \(\mathscr{S}_1(n,m,s) \subseteq \mathscr{S}_2(n,m,s)\), we have \(D^*_1 \in \mathscr{S}_2(n,m,s)\). Thus, by equation (\ref{kindclass}), it follows that \({^{\#}{}_{1}^{m_1}C_2^{(0)}(D^*_1)} \leq {^{\#}{}_{1}^{m_2}C_2^{(0)}(D^*_1)}\) and \({^{\#}{}_{2}^{\phi_1}C_2^{(0)}(D^*_1)} \leq {^{\#}{}_{2}^{\phi_2}C_2^{(0)}(D^*_1)}\). This means that both \({^{\#}{}_{1}^{m_1}C_2^{(0)}(D^*_1)}\) and \({^{\#}{}_{2}^{\phi_1}C_2^{(0)}(D^*_1)}\) for the \(\text{B}^1\text{-GMC}\) design are no greater than the respective maximum possible values within \(\mathscr{S}_2(n,m,s)\), with both inequalities holding simultaneously when the resolution \(R \geq IV\). Combining this with Theorem \ref{th3} and Remark \ref{re3}, part (ii) is proved.\hfill $\qed$
\end{proof}

Cheng and Mukerjee (\citeyear{mukerjee_blocked_2001}) proposed the MEC criterion for blocked designs.  Let \(E_r(D)\) be the number of models in design \(D\) that can estimate all main effects and \(r\) 2fic's simultaneously. Let \(f\) be the total number of \(T_l\)'s in the \(\phi\)-class of the B-WDM.  Let \(\left| \mathcal{C}_i \right|\) denote the number of alias sets containing exactly \(i\) 2fic's. Obviously, \(|\mathcal{C}_i| = {^{\#}{}_{2}^{\phi}C_2^{(i-1)}}/i\). Denote \(l = \min\{n(n-1)/2,\; s^{m}\}\). For the $Kind~2$ case of MBV designs, a simple calculation of \(E_r(D)\) can be achieved using the classification model of the \(\text{B}^2\text{-ACNP}\).

\begin{theorem}\label{th5}
	For any \(s^{n-m}:s^p\) MBV design \(D=(D_t: D_b)\) having resolution \(R \geq III\) and no main effect aliased with significant block components,  for \(r \leq f\) we have
	\begin{equation}\label{ErD}
	E_r(D) = \sum_{\substack{r_1 + r_2 + \cdots + r_l = r \\ 0 \le r_k \le \left| \mathcal{C}_i \right|}} \prod_{i=1}^l i^{r_i} \binom{\left| \mathcal{C}_i \right|}{r_i} = \sum_{1\leq {\phi l}_1 < {\phi l}_2 < \cdots < {\phi l}_r\leq N}
	\prod_{q=1}^r F_2(T_{{\phi l}_q}, 2),
	\end{equation}
	and for \(r > f\), \(E_r(D) = 0\).
\end{theorem}
\begin{proof}
	Note that all main effects belong to the \(m\)-class and are not aliased with each other. Thus, all main effects are estimable, whereas a 2fic is possibly estimable only when it belongs to the \(\phi\)-class.
	When \(r > f\), there obviously exists no model satisfying the requirement, so \(E_r(D) = 0\).
	When \(r \leq f\), by Lemma 1 of Li et al. (2015), the number of 2fic's in any alias set is at most \(l\)  in a MBV design \(D\). Hence, \(i = 1, 2, \dots, l\). The number of ways to select \(r_i\) 2fic's from these \(\left| \mathcal{C}_i \right|\) alias sets (at most one per alias set) is \(i^{r_i} \binom{\left| \mathcal{C}_i \right|}{r_i}\). For a given set of \(r_1, \dots, r_l\) satisfying \(\sum r_i = r\), the total number of ways to select the alias sets across all classes is \(\prod_{i=1}^l i^{r_i} \binom{\left| \mathcal{C}_i \right|}{r_i}\). Summing over all possible combinations of \(r_1, \dots, r_l\) gives the left-hand side of equation (\ref{ErD}).
	
    The right-hand side can be proved similarly. All estimable 2fic's lie in the \(f\) alias sets corresponding to \(T_{\phi l}\). Select \(r\) effect distribution vectors \(T_{{\phi l}_1}, T_{{\phi l}_2}, \dots, T_{{\phi l}_r}\) from the \(\phi\)-class, and  one 2fic from each of these \(r\) \(T_{\phi l}\)'s. Then, the total number of possible models is \(\prod_{q = 1}^r F_2(T_{{\phi l}_q}, 2)\). Summing this product over all choices of the \(T_{\phi l}\)'s yields the right-hand side of equation (\ref{ErD}).\hfill $\qed$
\end{proof}

\begin{corollary}
	For any \(s^{n-m}:s^p\) MBV design \(D=(D_t:D_b)\) and \(l = \min\{n(n-1)/2,\; s^{m}\}\), we have
	$$
	 \sum_{i=1}^{l} i\left| \mathcal{C}_i \right|=\sum_{i=1}^{l}{^{\#}{}_{2}^{\phi}C_2^{(i-1)}}=\sum_{q=1}^{f}F_2(T_{{\phi l}_q}, 2),\quad 
	 \sum_{i=1}^{l} i^2\left| \mathcal{C}_i \right|=\sum_{i=1}^{l} i{^{\#}{}_{2}^{\phi}C_2^{(i-1)}}=\sum_{q=1}^{f}(F_2(T_{{\phi l}_q}, 2))^2.  
	$$
\end{corollary}

\begin{example}\label{exm4}
	(i)  Consider a \(2^{5-1}:2^2\) \(\text{B}^1\text{-GMC}\) design \(D\). Since \(p=2\), by Theorem \ref{th4}, it is also a \(\text{B}^2\text{-GMC}\) design. Based on the results of Wei et al. (2014), the defining relation of the \(\text{B}^2\text{-GMC}\) design with resolution \(R=III\) is \(I = 12345 = 12b_1 = 234b_2\). After calculating its alias sets, the corresponding B-WDM is obtained as
	$$
	T=[T_1,T_2,\cdots,T_{15}]=\left[
	\begin{array}{ccccccccccccccc}
	1 & 1 & 2 & 1 & 2 & 2 & 3 & 1 & 2 & 2 & 3 & 2 & 3 & 3 & 4\\
	4 & 4 & 3 & 4 & 3 & 3 & 2 & 4 & 3 & 3 & 2 & 3 & 2 & 2 & 1
	\end{array}\right].
	$$
	It is partitioned into four classes, where the subset \(\{T_5, T_6, T_7, T_9, T_{10}, T_{11}, T_{12}\}\) constitutes the \(\phi\)-class \(T_{\phi l}\). Based on the set \(\{F_i(T_l, j), i, j = 0, 1, \cdots, n\}\) and applying Theorems \ref{th1}–\ref{th5}, we calculate \((A_1, A_2, \cdots, A_5) = (0, 0, 0, 0, 1)\) and \((B_1, B_2, \cdots, B_5) = (0, 3, 3, 0, 0)\). Furthermore, we obtain \(C_1(D) = 5\) and \(C_2(D) = 7\). Since \(f = 7\), we have \(E_1(D) = 7\), \(E_2(D) = 21\), \(E_3(D) = 35\), \(E_4(D) = 35\), \(E_5(D) = 21\), \(E_6(D) = 7\), and \(E_7(D) = 1\).
	
	(ii) Consider a \(3^{6-2}:3^3\) MBV design \(D\) with resolution \(R=III\). Its defining relation is \(I = 1235^2 = 12^246^2 = 12b_1 = 134b_2 = 23^24^2b_3\). After calculating all its alias sets, the corresponding B-WDM is obtained. Following classification and in conjunction with the set \(\{F_i(T_l, j), i, j = 0, 1, \cdots, n\}\), Theorem \ref{th1}, and Theorem \ref{th2}, we calculate \((A_1, A_2, \cdots, A_6) = (0, 0, 0, 2, 2, 0)\) and \((B_1, B_2, \cdots, B_6) = (0, 2, 12, 10, 8, 4)\). From Theorem \ref{th3}, we obtain \(C_1(D) = 6\) and \(CC(D) = 18\). Since it contains 30 \(T_l\)'s belonging to the \(\phi\)-class, we have \(f = 30\). According to Theorem \ref{th5}, we can calculate \(E_r(D)\); for example, \(E_1(D) = 28\), \(E_2(D) = 373\), etc. 
\end{example}

\section{Algorithms and visualization of $\text{B}^2\text{-ACNP}$}\label{sec4}

\subsection{The $\text{B}^2\text{-ACNP}$ algorithms}
In this subsection, we present a concise algorithm for computing the complete \(\text{B}^2\text{-ACNP}\) of \(s^{n-m}:s^{p}\) MBV designs. The procedure consists of three main steps: (i) construct the defining contrast subgroup \(G = G_t: G_b\); (ii) generate all alias sets and the B‑WDM, and classify the vectors \(T_{*l}\) into the four categories \(g, b, m, \phi\); (iii) accumulate the counts $^{\#}{}_{i}^*C_j^{(k)}$ from the B‑WDM. The algorithm is implemented in Python, and its computational complexity primarily depends on the number of independent columns \(q = n - m\) and the number of blocks \(p\). The detailed process of Algorithm \ref{alg1} is described by the following three steps:

{\it Step 1: Construct  the defining contrast subgroup \(G\).} 
Let \(V(I_t)\) be an \(n \times m\) matrix whose columns are \(m\) independent treatment defining words. All linear combinations of treatment defining words over \(\mathrm{GF}(s)\) are generated by an \(m \times s^m\) coefficient matrix \(Y_t\). The matrix $V_1(G_t) = V(I_t) \, Y_t \pmod s
$ contains all treatment-defining words and the equivalent forms. After removing the columns whose first nonzero element is not 1, we obtain \(V_1(G_t) = [\mathbf{0} \mid V(G_t)]\).
For the blocking part, let \(V(I_b)\) be an \(n \times p\) matrix whose columns are \(p\) independent block defining words. The coefficient matrix \(Y_b\) generates all linear combinations of block defining words over \(\mathrm{GF}(s)\). Since we are only concerned with block main effects and 2bic's, \(Y_b\) retains only those columns with exactly one or two nonzero elements. Compute $V_1(G_b) = V(I_b) \, Y_b \pmod s
$ and remove duplicate columns to obtain the block components. Add each column of \(V_1(G_b)\) (mod \(s\)) to all columns of \(V_1(G_t)\) and obtain the block defining contrast subgroup \(V_2(G_b)\).

{\it Step 2: Calculate alias sets and the B-WDM.} All treatment components are represented by the first \(N = (s^{n-m} - 1)/(s - 1)\) columns of the \(s\)-level Yates order, arranged into an \(n \times N\) matrix \(Y_q\). Let \(V_2(G_t) = [\,\mathbf{0}\mid V(G_t)\mid 2V(G_t)\mid\cdots\mid(s-1)V(G_t)\,]\pmod s\) be the horizontally concatenated matrix of scaled treatment defining words. For each column \(\mathbf{y}_q\) of \(Y_q\), adding \(\mathbf{y}_q\) (mod \(s\)) to each column of \(V_2(G_t)\) yields the alias set containing the component represented by that \(\mathbf{y}_q\). This alias set consists of \(s^m\) component vectors; counting the number of nonzero entries in each vector gives the wordlength distribution vector \(T_l\) of the alias set. Repeating this process for all \(\mathbf{y}_q\) yields the complete B‑WDM \(T = (T_1, \dots, T_N)\). In particular, \(T_0\) corresponds to the defining contrast subgroup \(G_t\). 
The vectors \(T_l\) are then classified into the following four categories: 
(i) \(g\)-class: \(T_0\) (i.e., \(G_t\));
(ii) \(b\)-class: any \(T_l\) whose generating alias set contains a block component from \(V_2(G_b)\);
(iii) \(m\)-class: any \(T_l\) that contains the element \(1\) (i.e., a main effect);
(iv) \(\phi\)-class: all remaining \(T_l\).

{\it Step 3: Compute the complete B\(^2\)-ACNP.}
For each \(T_{*l}\) (\(* \in \{g,b,m,\phi\}\)), let \(p_i\) be the number of \(i\) elements in \(T_{*l}\) and \(q_j\) be the number of \(j\) elements. For a fixed pair \((i,j)\), update \(\#_i^{*}C_j^{(k)}\) as follows:
 If \(i = j\), set \(k = q_j - 1\) and add \(p_i\) to \(^{\#}{}_{i}^*C_j^{(k)}\); 
if \(i \neq j\), set \(k = q_j\) and add \(p_i\) to \(^{\#}{}_{i}^*C_j^{(k)}\).

After processing all \(T_{*l}\), the complete B\(^2\)-ACNP is obtained. Since most \(^{\#}{}_{i}^*C_j^{(k)}\) are zero, a sparse storage structure (e.g., Python’s ``defaultdict”) should be used for efficiency. 

\begin{algorithm}
	\caption{Complete computation of the B$^2$-ACNP}\label{alg1}
	\KwIn{\\
		$\quad$$n$, $p$: The number of factors and block factors.\\
		$\quad$$s$: Factorial levels.\\
		$\quad$$V(I_t)$, $V(I_b)$: Matrices of independent defining words.
		\\}
	\KwOut{\\
		All  of  $^{\#}{}_{i}^*C_j^{(k)}$.
		\\}
	m = The number of columns in matrix \( V(I_t) \)\;
	Generate coefficient matrices $Y_t$ and $Y_b$ \;
	$V_1(G_t) \gets V(I_t) \cdot Y_t \pmod s$, \quad $V_1(G_b) \gets V(I_b) \cdot Y_b \pmod s$\;
	Remove the column vectors in \(V_1(G_t)\) and \(V_1(G_b)\) whose first nonzero element is not 1\; 
	For each column of \(V_1(G_b)\), add it \(\pmod s\) to each column of \(V_1(G_t)\), and store the results into \(V_2(G_b)\)\;
	Compute the number of independent columns \(q = n - m\) and \(N = (s^q - 1)/(s - 1)\)\;
	Generate the \(n \times N\) Yates matrix \(Y_q\), i.e., \(H_q\)\;
	$V(G_t) \gets V_1(G_t)$ with the first (identity) column removed\;
	$V_2(G_t) \gets \bigl[ \mathbf{0} \mid V(G_t) \mid 2V(G_t) \mid \cdots \mid (s-1)V(G_t) \bigr] \pmod s$\;
	$T_0 \gets$ wordlengths of columns of $V_1(G_t)$ (count of non-zero entries per column)\;
	\For{each column $\mathbf{y}_q$ of $Y_q$}{
		$A_l \gets \mathbf{y}_q$ (mod $s$) added to each column of $V_2(G_t)$\;
		$T_l \gets$ number of non-zero entries in each column of $A_l$\;
	}
	Classify \(T_l\) into four categories: \(g, b, m, \phi\)\;
	\For{all \(l\) in each class \(*\) (\(* \in \{g,b,m,\phi\}\))}
	{
		\For{$i=0 \to n$}
		{
			\For{$j=0 \to n$}
			{
				Calculate \( p_i \) and \( q_j \) in vector \( T_l \)\;
				\eIf{$i==j$}
				{
					$k=q_j-1$\;
					${^{\#}{}_{i}^*C_j^{(k)}} \leftarrow {^{\#}{}_{i}^*C_j^{(k)}} + p_i$\;
				}
				{
					$k=q_j$\;	
					${^{\#}{}_{i}^*C_j^{(k)}} \leftarrow {^{\#}{}_{i}^*C_j^{(k)}} + p_i$\;
				}
			}
		}
	}
	Output all ${^{\#}{}_{i}^*C_j}$
\end{algorithm}

\begin{example}\label{exm5}
	Consider a $3^{6-2}:3^3$ design $D$ with $I=12^235^2= 12^246^2=12b_1=13b_2=14b_3$. By Algorithm \ref{alg1}, we have
	$$
	V(I_t)=\begin{bmatrix}
	1 & 1 \\2 & 2 \\1 & 0 \\0 & 1 \\2 & 0\\0 & 2 \end{bmatrix}_{6\times 2}, 
	V(I_b)=\begin{bmatrix}
	1 & 1 & 1\\1 & 0 & 0\\0 & 1 & 0\\0 & 0 & 1\\0 & 0 & 0\\0 & 0 & 0\end{bmatrix}_{6\times 3}, 
	Y_t= \begin{bmatrix}
	0 & 1 & 2 & 0 & 1 & 2 & 0 & 1 & 2 \\
	0 & 0 & 0 & 1 & 1 & 1 & 2 & 2 & 2
	\end{bmatrix}
	_{2\times 9},
	Y_b= \begin{bmatrix}
	1 & 2 & 0 & \cdots & 0 & 0 \\
	0 & 0 & 1 & \cdots & 1 & 2 \\
	0 & 0 & 0 & \cdots & 2 & 2 \\
	\end{bmatrix}
	_{3\times 18}.
	$$
	Through modulo 3 matrix multiplication, we obtain \( V_1(G_t) = V(I_t) \times Y_t \) and \( V_1(G_b) = V(I_b) \times Y_b \). By screening the column vectors of the two matrices respectively, we obtain the treatment defining contrast subgroup \(G_t\) and \(U(D_b)\):
	$$
	V_1(G_t)= 
	\begin{bmatrix}
	0 & 0 & 1 & 1 & 1 \\
	0 & 0 & 2 & 2 & 2 \\
	0 & 1 & 2 & 1 & 0 \\
	0 & 2 & 2 & 0 & 1 \\
	0 & 2 & 1 & 2 & 0 \\
	0 & 1 & 1 & 0 & 2 
	\end{bmatrix}
	_{6\times 5},
	V_1(G_b) =
	\left[
	\begin{array}{ccccccccc}
	1 & 1 & 1 & 0 & 1 & 0 & 0 & 1 & 1 \\
	1 & 0 & 0 & 0 & 0 & 1 & 1 & 2 & 2 \\
	0 & 1 & 0 & 1 & 2 & 0 & 2 & 0 & 2 \\
	0 & 0 & 1 & 2 & 2 & 2 & 0 & 2 & 0 \\
	0 & 0 & 0 & 0 & 0 & 0 & 0 & 0 & 0 \\
	0 & 0 & 0 & 0 & 0 & 0 & 0 & 0 & 0
	\end{array}
	\right]_{6\times 9}
	$$ 
	Thus, the defining contrast subgroup for treatment components is \( G_t = \{I, 34^25^26, 12^23^24^256, 12^235^2, 12^246^2\} \), and \( U(D_b)= \{12, 13, 14, 34^2, 13^24^2, 24^2, 23^2, 12^24^2, 12^23^2\} \).  Algorithm \ref{alg1} then yields the defining contrast subgroup \( G_b \) for the MBV design. Subsequently, \( V_2(G_t) \) and \( Y_q \) are obtained from  the algorithm. Each component is selected from \( Y_q \) one at a time, and its modulo 3 addition with each column vector of \( V_2(G_t) \) is computed. For example, when \( \mathbf{y_q} = (1, 0, 0, 0, 0, 0)^\top \), the resulting alias set matrix \( A_1 \) is:
	$$
	A_1= 
	\begin{bmatrix}
	1 & 1 & 2 & 2 & 2 & 1 & 0 & 0 & 0 \\
	0 & 0 & 2 & 2 & 2 & 0 & 1 & 1 & 1 \\
	0 & 1 & 2 & 1 & 0 & 2 & 1 & 2 & 0 \\
	0 & 2 & 2 & 0 & 1 & 1 & 1 & 0 & 2 \\
	0 & 2 & 1 & 2 & 0 & 1 & 2 & 1 & 0 \\
	0 & 1 & 1 & 0 & 2 & 2 & 2 & 0 & 1
	\end{bmatrix}
	_{6\times 9}.
	$$
	Hence, the corresponding alias set is obtained as \(1 = 134^25^26 = 1^22^23^24^256 = 1^22^235^2 = 1^22^246^2 = 13^2456^2 = 2345^26^2 = 23^25 = 24^26\), and the corresponding vector is \(T_1 = (1, 5, 6, 4, 4, 5, 5, 3, 3)^\top\). Repeating this process over the entire \(Y_q\) matrix allows for the construction of the complete B-WDM.
	After classifying \( T_l \) into the four categories \( g, b, m, \) and \( \phi \), the complete results of the \(\text{B}^2\text{-ACNP}\) are computed according to the algorithm and listed in Table \ref{tab1}. For brevity,  all \(^{\#}{}_{i}^*C_j^{(k)}\) are zero are omitted from the table.
\end{example}

\begin{table}[htbp]
	\centering
	\caption{Complete B$^2$-ACNP of the $3^{6-2}:3^3$ design.}\label{tab1}
	\begin{tabular*}{\hsize}{@{}@{\extracolsep{\fill}}llllllllll@{}}
		\hline 
		Class&\( ^{\#}{}_{i}^*C_j^{(k)} \)& $j=0$ & $j=1$ & $j=2$ & $j=3$ & $j=4$ & $j=5$ & $j=6$ \\ 
		\hline 
		&$i=0$ & 1 & 1 & 1 & 1 & $0^3$, 1 & 1 & 0, 1\\ 
		$g-$class  &$i=4$ & 0, 3 & 3 & 3 & 3 & $0^2$, 3 & 3 & 0, 3\\  
		&$i=6$ & 0, 1 & 1 & 1 & 1 & $0^3$, 1 & 1 & 1\\   
		\hline
		&$i=2$ & 13 & 13 & 3, 10 & 10, $0^{3}$, 3 & 3, $0^{2}$, 10 & $0^{2}$, 3, 0, 10 & 10, 0, 3\\
		&$i=3$ & 14 & 14 & 2, 12 & 0, 2, 0, 12 & 12, $0^{5}$, 2 & 2, 0, 12 & 0, 2, 12\\ 
		$b-$class &$i=4$ & 21 & 21 & 6, 0, 15 & 15, 0, 6 & $0^{2}$, 15, $0^{2}$, 6 & 6, $0^{3}$, 15 & 15, 6\\  
		&$i=5$ & 26 & 26 & 0, 6, 20 & 20, $0^{3}$, 6 & 6, $0^{2}$, 20 & 0, 6, 0, 20 & 20, 0, 6\\
		&$i=6$ & 7 & 7 & 1, 6 & $0^{2}$, 1, 0, 6 & 6, $0^{5}$, 1 & 1, 0, 6 & 1, 6\\  
		\hline
		&$i=1$ & 6 & 6 & 6 & $0^{2}$, 6 & $0^{2}$, 6 & $0^{3}$, 6 & 0, 6\\ 
		&$i=3$ & 12 & 0, 12 & 12 & 0, 12 & $0^{2}$, 12 & $0^{3}$, 12 & 0, 12\\
		$m-$class &$i=4$ & 12 & 0, 12 & 12 & $0^{2}$, 12 & 0, 12 & $0^{3}$, 12 & 0, 12\\  
		&$i=5$ & 18 & 0, 18 & 18 & $0^{2}$, 18 & $0^{2}$, 18 & $0^{2}$, 18 & 0, 18\\
		&$i=6$ & 6 & 0, 6 & 6 & $0^{2}$, 6 & $0^{2}$, 6 & $0^{3}$, 6 & 6\\    
		\hline
		&$i=2$ & 17 & 17 & 12, 2, 3 & 5, 0, 12 & $0^3$, 17 & 3, 0, 12, 0, 2 & 2, 12, 0, 3\\
		&$i=3$ & 54 & 54 & 30, 24 & 0, 30, 24 & $0^3$, 48, $0^2$, 6 & 6, 0, 24, 24 & 24, 30\\ 
		$\phi-$class &$i=4$ & 84 & 84 & 42, 36, 3, 3 & 6, 0, 54, 24 & $0^2$, 66, $0^2$, 18 & 21, 0, 36, 24, 3 & 27, 54, 0, 3\\  
		&$i=5$ & 52 & 52 & 24, 24, 4 &4, 0, 24, 24 & $0^3$, 52 & 0, 24, 24, 4 & 28, 24\\
		&$i=6$ & 18 & 18 & 3, 12, 0, 3 & 3, 0, 15 & $0^3$, 15, $0^2$, 3 & 6, 0, 12 & 15, 0, 3\\  
		\hline
	\end{tabular*}
	\label{tab1}
\end{table}

\subsection{Visualization of the $\text{B}^2\text{-ACNP}$}
The B\(^2\)-ACNP is closely related with \((i, j, k, {^{\#}{}_{i}^*C_j^{(k)}})\). For ease of interpretation, we generate three-dimensional scatter plots and use a color map to represent the magnitude of \({^{\#}{}_{i}^*C_j^{(k)}}\). Projections onto the \(ij\)-plane, \(ik\)-plane, and \(jk\)-plane are also provided. The degree of aliasing is highest when the component order is close to \((s-1)n/s\). This phenomenon can be explained by the total number of \(i\)th-order components, \(N_i = \binom{n}{i}(s-1)^{i-1}\). Consider the ratio
\[
R_i = \frac{N_i}{N_{i-1}} = \frac{n-i+1}{i}(s-1).
\]
When \( R_i > 1 \), \( N_i \) increases with \( i \); when \( R_i < 1 \), \( N_i \) decreases. Setting \( R_i \approx 1 \) yields \( i \approx (s-1)n/s \), the order at which the number of components is maximized, and where the potential for aliasing tends to be largest.

\begin{example}\label{exm6}
	For a \(3^{9-4}:3^{2}\) design with the defining relations  
	\(I = 123456^{2} = 12^{2}3^{2}47^{2} = 12^{2}4^{2}58^{2} = 23^{2}4^{2}59^{2} = 123^{2}5^{2}b_{1} = 134^{2}5^{2}b_{2}\),  
	its aliasing pattern is shown in Fig. \ref{fig5}. As predicted by the formula for \(s=3\), the ``peak" of aliasing occurs near \(i \approx 2n/3\). The main plot and the \(jk\)-plane projection exhibit a unimodal arch-shaped distribution.
\end{example}

These visualizations clearly illustrate aliasing among the four component classes \(g, b, m, \phi\), facilitating an intuitive understanding of the \(\text{B}^2\text{-ACNP}\) and enabling comparison of aliasing severity across different designs.

\begin{figure}[!hpt]
	\centering 
	
	\begin{subfigure}[b]{0.48\textwidth}
		\centering
		\includegraphics[width=\textwidth]{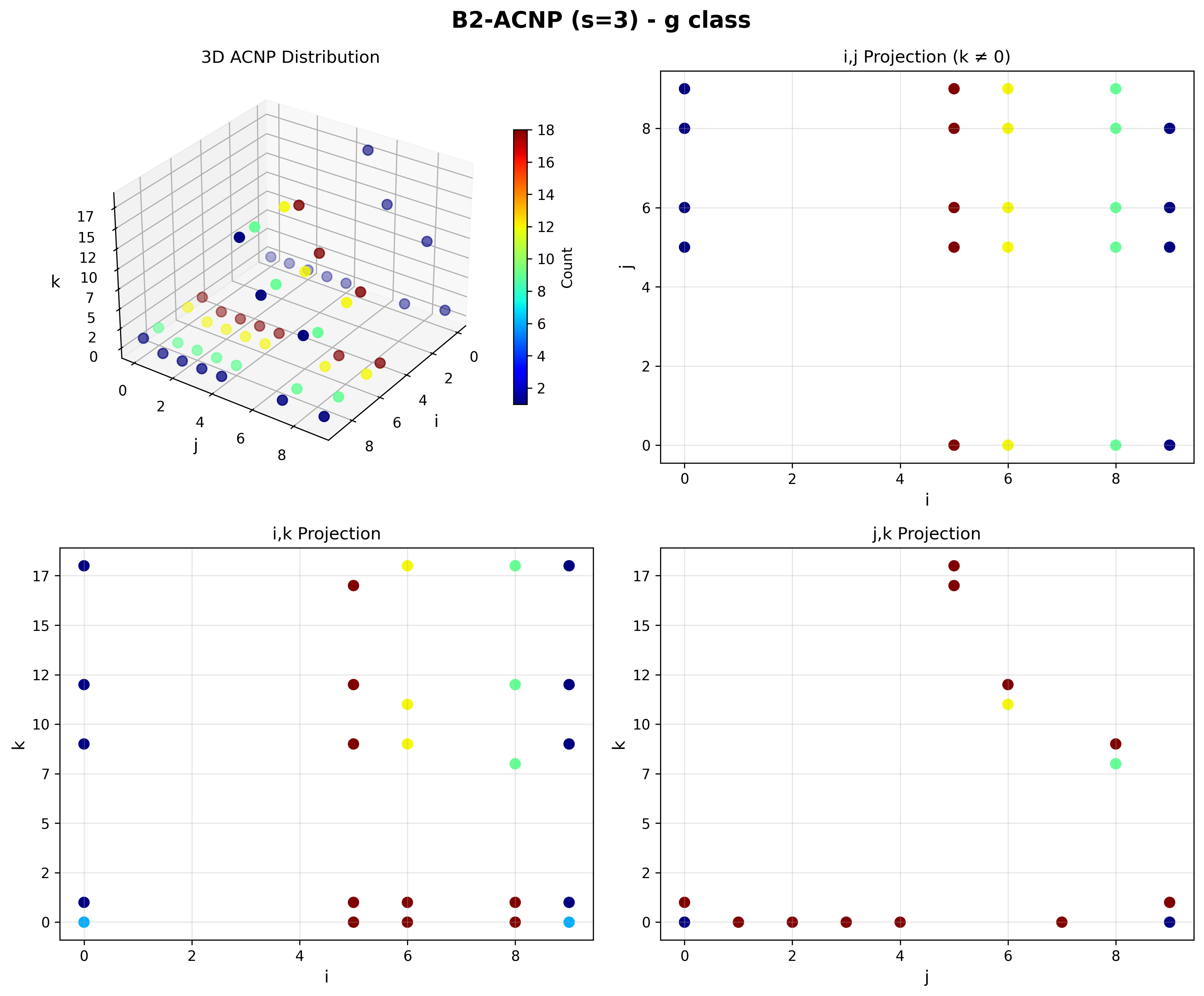} 
		\caption{$\text{B}^2\text{-ACNP}$($g-$class)}
		\label{fig1} 
	\end{subfigure}
	\hfill
	\begin{subfigure}[b]{0.48\textwidth}
		\centering 
		\includegraphics[width=\textwidth]{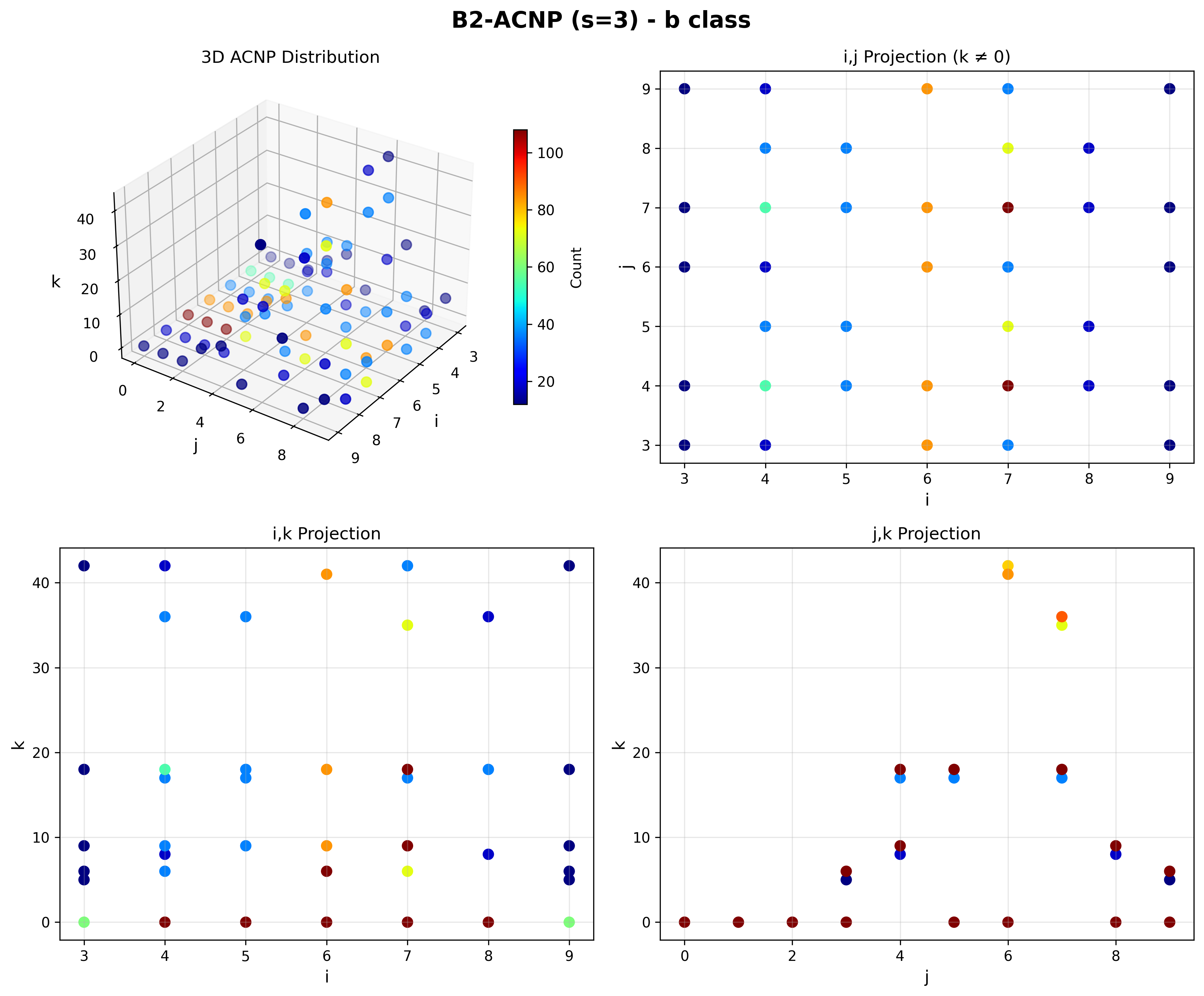} 
		\caption{ $\text{B}^2\text{-ACNP}$($b-$class) }
		\label{fig2} 
	\end{subfigure}
	
	\begin{subfigure}[b]{0.48\textwidth}
		\centering
		\includegraphics[width=\textwidth]{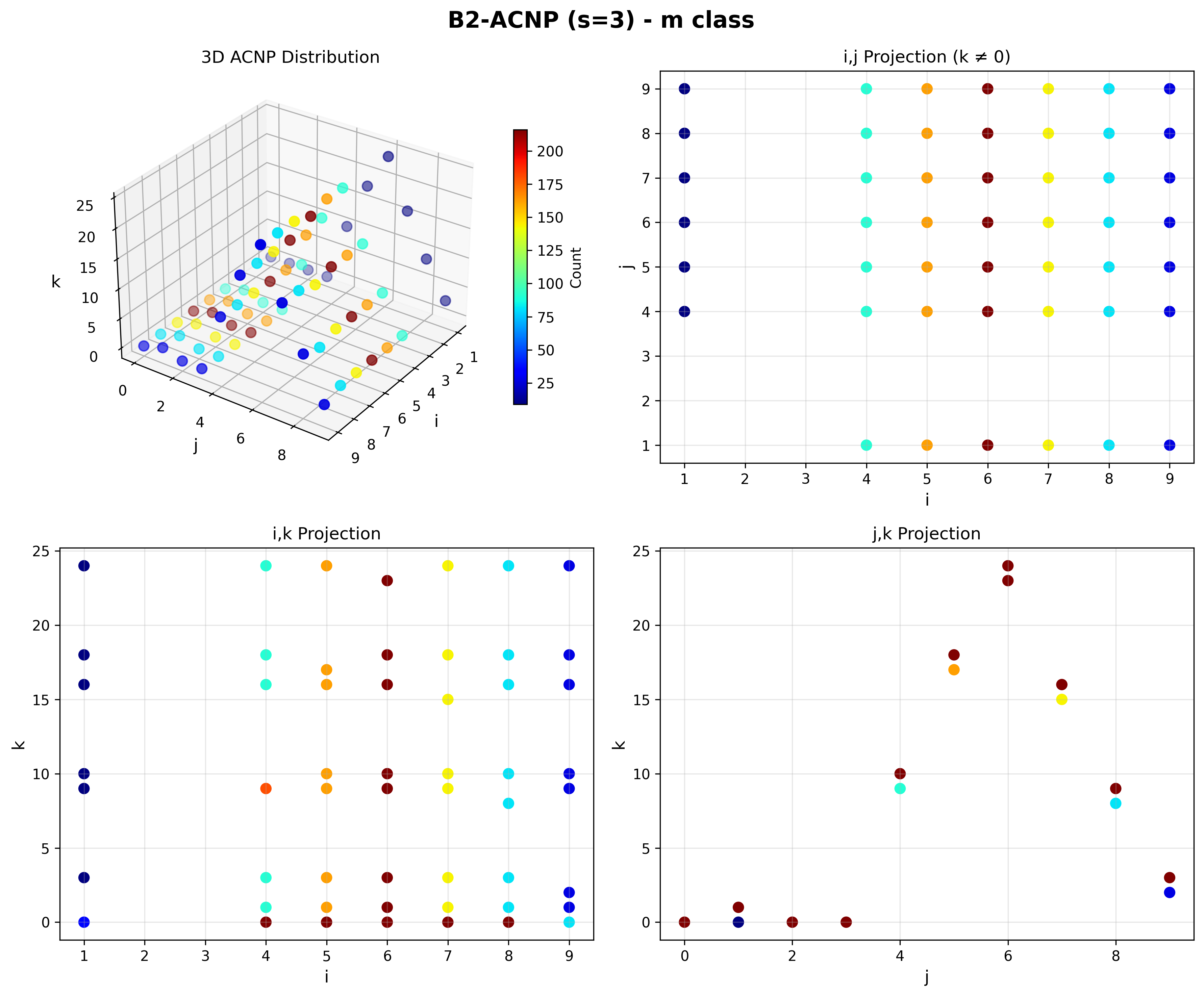} 
		\caption{$\text{B}^2\text{-ACNP}$($m-$class)}
		\label{fig3} 
	\end{subfigure}
	\hfill
	\begin{subfigure}[b]{0.48\textwidth}
		\centering 
		\includegraphics[width=\textwidth]{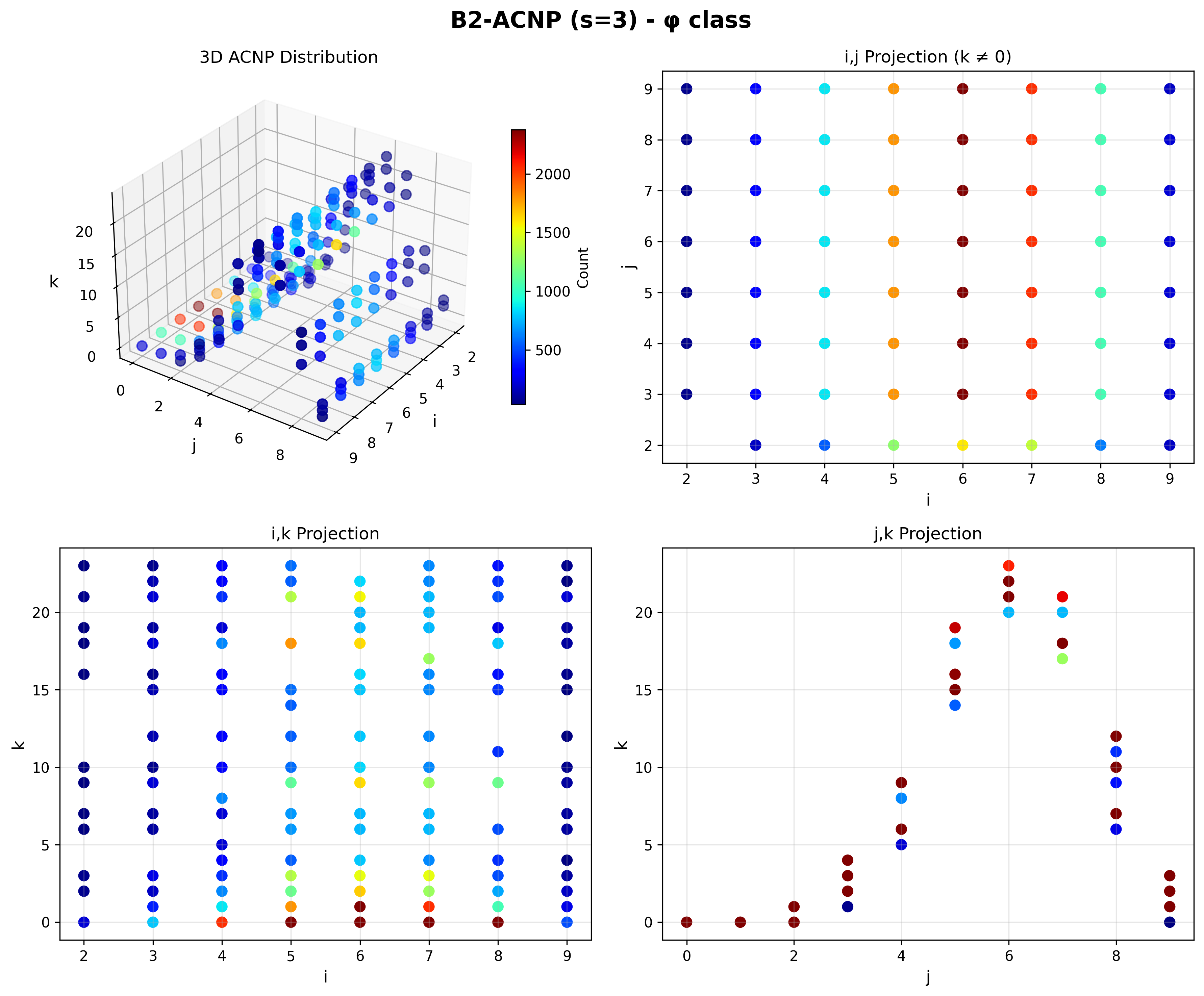} 
		\caption{ $\text{B}^2\text{-ACNP}$($\phi-$class) }
		\label{fig4} 
	\end{subfigure}
	\caption{Confounding pattern for three-level designs.}
	\label{fig5}
\end{figure}

\section{Applications of the \(\text{B}^2\text{-ACNP}\)}\label{sec5}
In this section, we first apply the \(\text{B}^2\text{-ACNP}\) to search \(\text{B}^2\text{-GMC}\) designs and then two real studies. Based on the computational algorithm for the \(\text{B}^2\text{-ACNP}\) proposed in Section \ref{sec4}, the \(\text{B}^2\text{-GMC}\) designs can be systematically searched for using a search algorithm in conjunction with it. 

\textit{Searching \(\text{B}^2\text{-GMC}\) designs.} 
Consider a \(3^{6-2}:3^3\) MBV design. Let \(1,2,3,4\) denote the four independent columns. The remaining columns \(d_1, d_2\) and the block components \(b_1, b_2, b_3\) are then selected from the set \(H_{43} \setminus \{1,2,3,4\}\). The total number of possible combinations is \(\binom{36}{5}\). Two designs are considered isomorphic if their \(\text{B}^2\text{-ACNP}\) are exactly identical and they have the same ranking under the \(\text{B}^2\text{-GMC}\) criterion. If the focus is only on the aliasing structure between main effects and two-factor interactions, Table \ref{tab2} lists the top ten non-isomorphic designs ranked according to the \(\text{B}^2\text{-GMC}\) criterion.

\begin{table}[htbp]
	\centering
	\caption{Non-isomorphic $3^{6-2}:3^3$ designs.}
	\renewcommand{\arraystretch}{1.5} 
	\begin{tabular*}{\hsize}{@{}@{\extracolsep{\fill}}ccccc@{}}  
		\hline
		\multirow{2}{*}{Design} & \multirow{2}{*}{Additional columns} & \multirow{2}{*}{Block columns} & \multicolumn{2}{c}{ \(\text{B}^2\text{-ACNP}\)} \\  \cline{4-5}  
		& & & $^{\#}{}_{1}^mC_2$ & $^{\#}{}_{2}^\phi C_2$ \\
		\hline
		$3^{6-2}:3^3.1$ & $123,~12^24$ & $12,~134,~23^24^2$ & (6) & (18,10)\\
		$3^{6-2}:3^3.2$ & $123,~12^24$ & $13,~234,~123^24$ & (6) & (17,10)\\
		$3^{6-2}:3^3.3$ & $123,~12^24$ & $12,~13^24,~12^23^24^2$ & (6) & (17,8)\\
		$3^{6-2}:3^3.4$ & $123,~12^24$ & $12,~14,~24^2$ & (6) & (17,6)\\
		$3^{6-2}:3^3.5$ & $123,~12^24$ & $12^23,~234,~13^24$ & (6) & (16,12)\\
		$3^{6-2}:3^3.6$ & $123,~12^24$ & $12,~34,~1234$ & (6) & (16,10)\\
		$3^{6-2}:3^3.7$ & $123,~12^24$ & $12,~13,~23^2$ & (6) & (16,8)\\
		$3^{6-2}:3^3.8$ & $123,~12^24$ & $12,~14,~12^234^2$ & (6) & (16,6)\\
		$3^{6-2}:3^3.9$ & $123,~124$ & $12,~134,~23^24^2$ & (6) & (15,12)\\
		$3^{6-2}:3^3.10$ & $123,~12^24$ & $12,~12^23,~13^24$ & (6) & (15,8)\\
		\hline
	\end{tabular*}
	\label{tab2}
\end{table}
In fact, since the calculation of the complete \(\text{B}^2\text{-ACNP}\) has become feasible, researchers can, based on specific experimental requirements, emphasize specific \(^{\#}{}_{i}^*C_j^{(k)}\) and use them to perform custom rankings of candidate designs, thereby enhancing the flexibility of the criterion in practical applications. For example, if there is a need to specifically control the aliasing of block components, the \(^{\#}{}_{i}^bC_j^{(k)}\) can be included in the ranking basis.
The following provides examples of applications of the \(\text{B}^2\text{-GMC}\) criterion in practical scenarios.

 \textit{Experiment of truck leaf springs (Wu and Hamada (\citeyear{wu2011experiments})).} To fairly compare the effectiveness of the \(\text{B}^1\text{-GMC}\) and \(\text{B}^2\text{-GMC}\) criteria in MBV designs under $Kind~2$, this study employs a simulation method based on a five-factor two-level design on the truck leaf spring heat treatment process improvement experiment. The following regression model is established to generate simulated data:
	\[
	\begin{aligned}
	y = \ 7.636 + 0.1106A + 0.0881B - 0.1298C + 0.0519D 
	- 0.0423AC - 0.0827BC + \sum_{b=1}^{8} \gamma_b I_b + \varepsilon,
	\end{aligned}
	\]
	where \(A\) through \(E\) are the coded values of the factors (taking ±1), \(\gamma_b\) represents the random effect of the \(b\)th block, following a \(N(0, 0.5^2)\) distribution, and \(\varepsilon\) is the random error, following a \(N(0, 0.5^2)\) distribution. The model parameters are derived from a real \(2^5\) full factorial experiment and its regression analysis. In this example, block random effects \(\gamma_b\) are additionally introduced to simulate the potential variability that MBVs may introduce in practical experiments.
	
	Based on the search algorithms from Section \ref{sec4}, the optimal designs under the \(\text{B}^1\text{-GMC}\) and \(\text{B}^2\text{-GMC}\) criteria for the \(2^{5-1}:2^3\) MBV design were obtained, denoted as \(D_a\) and \(D_b\), respectively. Their defining relations are:
	\[
	\begin{aligned}
	D_a: I = ACDE = ABb_1 = BCb_2 = CDb_3, \quad
	D_b: I = ABCDE = ABb_1 = ACDb_2 = BCDb_3,
	\end{aligned}
	\]
	where \(b_1, b_2, b_3\) represent the three block main effects. These two design matrices were applied to the above model, and response data were generated via simulation. The complete design matrices are shown in Table \ref{tab3}.
	\begin{table}
		\centering
		\caption{Design matrix and response data for heat treatment simulation.}
        \begin{tabular*}{\hsize}{@{}@{\extracolsep{\fill}}ccrrrrrcrrrrrc@{}}
			\hline 
			Run & Block & $A_1$ & $B_1$ & $C_1$ & $D_1$ & $E_1$ & Readout &  $A_2$ & $B_2$ & $C_2$ & $D_2$ & $E_2$ & Readout  \\ 
			\hline 
			1 & 1 & 1 & -1 & 1 & -1 & -1 & 7.13 & -1 & -1 & -1 & 1 & -1 & 7.10 \\
			2 & 1 & -1 & 1 & -1 & 1 & 1 & 6.62 & -1 & -1 & 1 & -1 & -1 & 5.99 \\ 
			3 & 2 & -1 & -1 & 1 & -1 & 1 & 7.52 & 1 & 1 & -1 & -1 & 1 & 8.18 \\ 
			4 & 2 & 1 & 1 & -1 & 1 & -1 & 8.28 & 1 & 1 & 1 & 1 & 1 & 7.94 \\ 
			5 & 3 & -1 & 1 & 1 & -1 & 1 & 7.44 & -1 & -1 & -1 & -1 & 1 & 7.61 \\ 
			6 & 3 & 1 & -1 & -1 & 1 & -1 & 7.37 & -1 & -1 & 1 & 1 & 1 & 7.05 \\ 
			7 & 4 & 1 & 1 & 1 & -1 & -1 & 7.62 & 1 & 1 & -1 & 1 & -1 & 8.06 \\ 
			8 & 4 & -1 & -1 & -1 & 1 & 1 & 7.25 & 1 & 1 & 1 & -1 & -1 & 7.28 \\ 
			9 & 5 & -1 & 1 & -1 & -1 & -1 & 8.28 & -1 & 1 & -1 & 1 & 1 & 8.38 \\ 
			10 & 5 & 1 & -1 & 1 & 1 & 1 & 6.69 & -1 & 1 & 1 & -1 & 1 & 6.29 \\ 
			11 & 6 & 1 & 1 & -1 & -1 & 1 & 7.99 & 1 & -1 & -1 & -1 & -1 & 7.65 \\ 
			12 & 6 & -1 & -1 & 1 & 1 & -1 & 8.05 & 1 & -1 & 1 & 1 & -1 & 8.35 \\ 
			13 & 7 & 1 & -1 & -1 & -1 & 1 & 8.15 & -1 & 1 & -1 & -1 & -1 & 8.35 \\ 
			14 & 7 & -1 & 1 & 1 & 1 & -1 & 7.28 & -1 & 1 & 1 & 1 & -1 & 7.28 \\ 
			15 & 8 & -1 & -1 & -1 & -1 & -1 & 7.54 & 1 & -1 & -1 & 1 & 1 & 7.78 \\ 
			16 & 8 & 1 & 1 & 1 & 1 & 1 & 7.30 & 1 & -1 & 1 & -1 & 1 & 7.19 \\  
			\hline 
		\end{tabular*}
		\label{tab3}
	\end{table}
    
	Based on the half-normal probability plots (Fig. \ref{fig25}), a preliminary screening for potentially significant treatment components was conducted for the two designs. Subsequently, regression analysis was performed on the screened effects at the 0.05 significance level, yielding regression models for designs \(D_a\) and \(D_b\). Their coefficients of determination were \(R_1^2 = 0.8849\) and \(R_2^2 = 0.9391\), respectively. However, in design \(D_a\), the treatment components \(AB\) and \(AC\) are aliased with block components. Consequently, its regression model is actually confounded with the block variables, making it unreliable to determine whether these two effects are truly significant:
	\[
	\begin{aligned}
	\hat{y}_1 &= 7.5319-0.1531C-0.1769D-0.1619E+0.1619AB-0.2281AC+0.2031ABD,\\
	\hat{y}_2 &= 7.5300+0.2738A+0.1900B-0.3587C+0.2125D+0.2450AC-0.1637BC+0.2712CD.
	\end{aligned}
	\]
	
	\begin{figure}[!hpt]
		\centering 
		\begin{subfigure}[b]{0.48\textwidth}
			\centering
			\includegraphics[width=\textwidth]{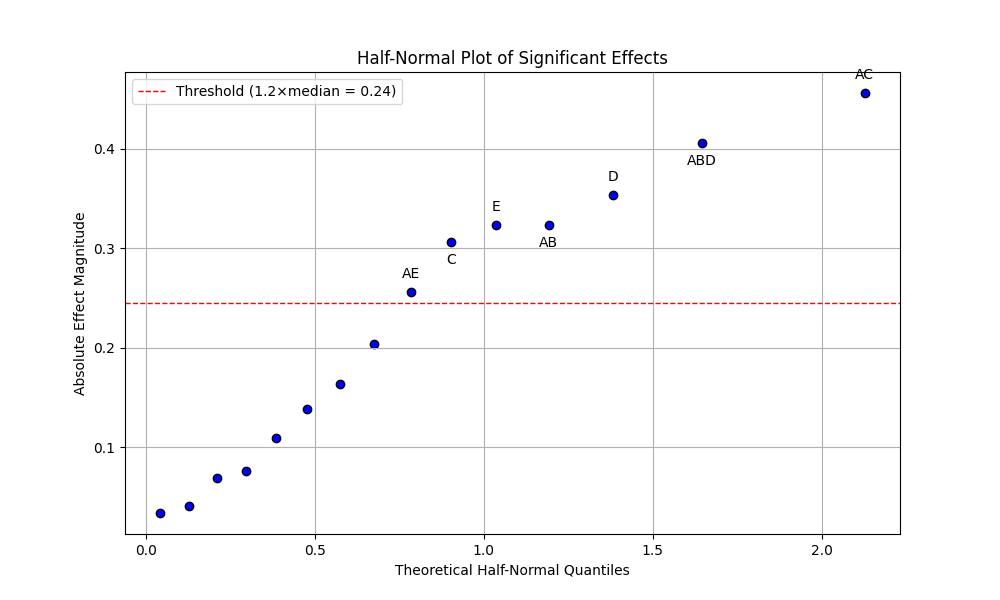} 
			\caption{Half-normal probability plot of \(D_a\)}
			\label{fig23} 
		\end{subfigure}
		\hfill
		\begin{subfigure}[b]{0.48\textwidth}
			\centering 
			\includegraphics[width=\textwidth]{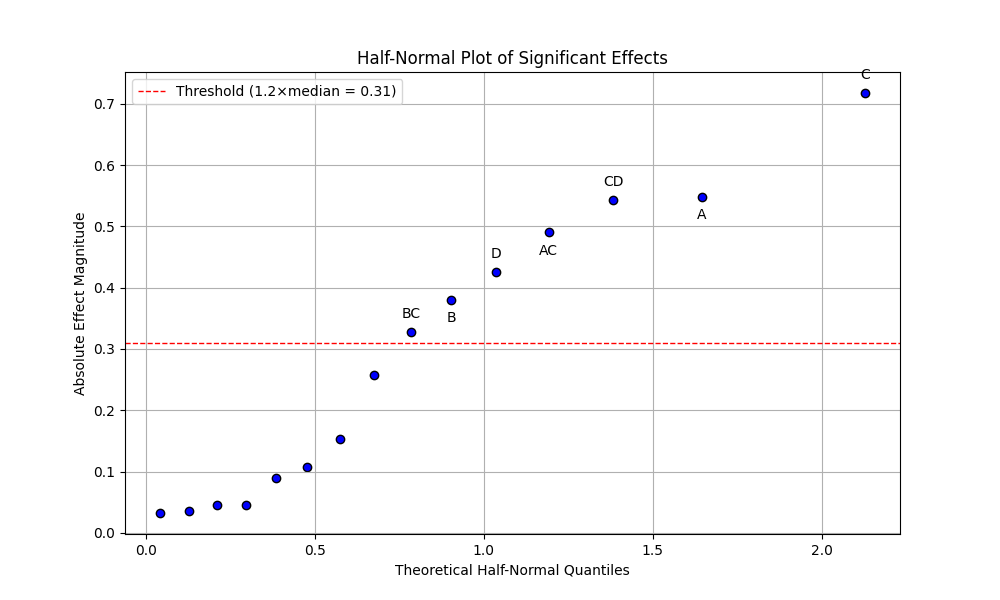} 
			\caption{ Half-normal probability plot of \(D_b\) }
			\label{fig24} 
		\end{subfigure}
		\caption{Half-normal probability plot for the experiment.}
		\label{fig25}
	\end{figure}
	
	The above results indicate that for experiments with MBVs, the optimal design selected by the \(\text{B}^2\text{-GMC}\) criterion exhibits superior model fitting capability. This difference can be explained by the aliasing structures of the two designs: The \(\text{B}^1\text{-ACNP}\) for \(D_a\) is \(^{\#}{}_{1}^mC_2 = (5), {^{\#}{}_{2}^\phi C_2} = (0)\), whereas the \(\text{B}^2\text{-ACNP}\) for \(D_b\) is \(^{\#}{}_{1}^mC_2 = (5), {^{\#}{}_{2}^\phi C_2} = (7)\). This implies that in design \(D_a\), all 2fis are confounded with block effects and thus nonestimable, while in design \(D_b\), only three 2fis are nonestimable. Consequently, the model fit for design \(D_a\) is significantly weaker than that for design \(D_b\).
	
	According to the conclusion of Zhang et al. (\citeyear{zhang_optimal_2011}), in $Kind~2$ two-level designs with MBVs, the number of clear main effects and clear 2fis in a \(\text{B}^2\text{-GMC}\) design is always greater than that in a \(\text{B}^1\text{-GMC}\) design. Combined with Theorem 4, this conclusion can be extended to high-resolution designs with MBVs at \(s\) levels. Therefore, in the context of MBVs, the model obtained from a \(\text{B}^2\text{-GMC}\) design is generally superior to that from a \(\text{B}^1\text{-GMC}\) design, and the simulation results of this study align with the theoretical expectation.

\textit{Drug experiments (Jaynes et al., \citeyear{jaynes_application_2013})}
	Consider an experiment investigating the effects of six antiviral drugs on Herpes simplex virus type 1 (HSV-1). The experiment involves six three-level factors: Interferon-alpha (A), Interferon-beta (B), Interferon-gamma (C), Ribavirin (D), Acyclovir (E), and TNF-alpha (F). Since the distribution of viral infection is skew-normal, \( y = \log_{10}(\text{readout}) \) is chosen as the response variable. The experiment is divided into three blocks, denoted by the block variable \( b_1 \).
	The experiment employs a \( 3^{6-2}:3^1 \) \(\text{B}^2\text{-GMC}\) blocked design with the defining relation \( I = ABCDE^2 = AB^2CF^2 = AC^2Db_1 \). The design matrix of this experiment and the corresponding data are presented in Table \ref{tab4}.

\begin{table}[htbp]
	\centering
\caption{Design matrix of the HSV-1 experiment.}\label{tab4}
	\footnotesize 
	\begin{tabular}{cccccccccccccccccc}
		\toprule
		\textbf{Run} & \textbf{A} & \textbf{B} & \textbf{C} & \textbf{D} & \textbf{E} & \textbf{F} & \textbf{Block} & \textbf{Readout} & 
		\textbf{Run} & \textbf{A} & \textbf{B} & \textbf{C} & \textbf{D} & \textbf{E} & \textbf{F} & \textbf{Block} & \textbf{Readout} \\
		\midrule
		1 & 0 & 0 & 0 & 0 & 0 & 0 & 0 & 49.1 & 42 & 1 & 2 & 1 & 1 & 2 & 0 & 1 & 2.3 \\
		2 & 0 & 1 & 0 & 0 & 1 & 2 & 0 & 37.2 & 43 & 2 & 0 & 2 & 1 & 2 & 1 & 1 & 14 \\
		3 & 0 & 2 & 0 & 0 & 2 & 1 & 0 & 39.9 & 44 & 2 & 1 & 2 & 1 & 0 & 0 & 1 & 4.1 \\
		4 & 1 & 0 & 1 & 0 & 2 & 2 & 0 & 28.7 & 45 & 2 & 2 & 2 & 1 & 1 & 2 & 1 & 2.5 \\
		5 & 1 & 1 & 1 & 0 & 0 & 1 & 0 & 41.8 & 46 & 2 & 0 & 0 & 2 & 1 & 2 & 1 & 2.3 \\
		6 & 1 & 2 & 1 & 0 & 1 & 0 & 0 & 50.6 & 47 & 2 & 1 & 0 & 2 & 2 & 1 & 1 & 1.3 \\
		7 & 2 & 0 & 2 & 0 & 1 & 1 & 0 & 24.4 & 48 & 2 & 2 & 0 & 2 & 0 & 0 & 1 & 2 \\
		8 & 2 & 1 & 2 & 0 & 2 & 0 & 0 & 18.2 & 49 & 0 & 0 & 1 & 2 & 0 & 1 & 1 & 3 \\
		9 & 2 & 2 & 2 & 0 & 0 & 2 & 0 & 24.3 & 50 & 0 & 1 & 1 & 2 & 1 & 0 & 1 & 1.6 \\
		10 & 2 & 0 & 0 & 1 & 0 & 2 & 0 & 23.7 & 51 & 0 & 2 & 1 & 2 & 2 & 2 & 1 & 0.9 \\
		11 & 2 & 1 & 0 & 1 & 1 & 1 & 0 & 6.2 & 52 & 1 & 0 & 2 & 2 & 2 & 0 & 1 & 1.6 \\
		12 & 2 & 2 & 0 & 1 & 2 & 0 & 0 & 6.2 & 53 & 1 & 1 & 2 & 2 & 0 & 2 & 1 & 3.9 \\
		13 & 0 & 0 & 1 & 1 & 2 & 1 & 0 & 6.8 & 54 & 1 & 2 & 2 & 2 & 1 & 1 & 1 & 0.8 \\
		14 & 0 & 1 & 1 & 1 & 0 & 0 & 0 & 8.9 & 55 & 2 & 0 & 0 & 0 & 2 & 2 & 2 & 14.5 \\
		15 & 0 & 2 & 1 & 1 & 1 & 2 & 0 & 7.6 & 56 & 2 & 1 & 0 & 0 & 0 & 1 & 2 & 24.5 \\
		16 & 1 & 0 & 2 & 1 & 1 & 0 & 0 & 5.2 & 57 & 2 & 2 & 0 & 0 & 1 & 0 & 2 & 30.2 \\
		17 & 1 & 1 & 2 & 1 & 2 & 2 & 0 & 6 & 58 & 0 & 0 & 1 & 0 & 1 & 1 & 2 & 50.5 \\
		18 & 1 & 2 & 2 & 1 & 0 & 1 & 0 & 7.6 & 59 & 0 & 1 & 1 & 0 & 2 & 0 & 2 & 34.1 \\
		19 & 1 & 0 & 0 & 2 & 0 & 1 & 0 & 4.2 & 60 & 0 & 2 & 1 & 0 & 0 & 2 & 2 & 39.4 \\
		20 & 1 & 1 & 0 & 2 & 1 & 0 & 0 & 2.4 & 61 & 1 & 0 & 2 & 0 & 0 & 0 & 2 & 31 \\
		21 & 1 & 2 & 0 & 2 & 2 & 2 & 0 & 1.8 & 62 & 1 & 1 & 2 & 0 & 1 & 2 & 2 & 17.2 \\
		22 & 2 & 0 & 1 & 2 & 2 & 0 & 0 & 4.4 & 63 & 1 & 2 & 2 & 0 & 2 & 1 & 2 & 11.8 \\
		23 & 2 & 1 & 1 & 2 & 0 & 2 & 0 & 3.7 & 64 & 1 & 0 & 0 & 1 & 2 & 1 & 2 & 4.8 \\
		24 & 2 & 2 & 1 & 2 & 1 & 1 & 0 & 2.7 & 65 & 1 & 1 & 0 & 1 & 0 & 0 & 2 & 4.6 \\
		25 & 0 & 0 & 2 & 2 & 1 & 2 & 0 & 3.5 & 66 & 1 & 2 & 0 & 1 & 1 & 2 & 2 & 4.5 \\
		26 & 0 & 1 & 2 & 2 & 2 & 1 & 0 & 1.4 & 67 & 2 & 0 & 1 & 1 & 1 & 0 & 2 & 3.3 \\
		27 & 0 & 2 & 2 & 2 & 0 & 0 & 0 & 2.5 & 68 & 2 & 1 & 1 & 1 & 2 & 2 & 2 & 3.7 \\
		28 & 1 & 0 & 0 & 0 & 1 & 1 & 1 & 11.8 & 69 & 2 & 2 & 1 & 1 & 0 & 1 & 2 & 2.5 \\
		29 & 1 & 1 & 0 & 0 & 2 & 0 & 1 & 14.5 & 70 & 0 & 0 & 2 & 1 & 0 & 2 & 2 & 14.2 \\
		30 & 1 & 2 & 0 & 0 & 0 & 2 & 1 & 14 & 71 & 0 & 1 & 2 & 1 & 1 & 1 & 2 & 2.7 \\
		31 & 2 & 0 & 1 & 0 & 0 & 0 & 1 & 19.3 & 72 & 0 & 2 & 2 & 1 & 2 & 0 & 2 & 2.5 \\
		32 & 2 & 1 & 1 & 0 & 1 & 2 & 1 & 8.6 & 73 & 0 & 0 & 0 & 2 & 2 & 0 & 2 & 1.6 \\
		33 & 2 & 2 & 1 & 0 & 2 & 1 & 1 & 10 & 74 & 0 & 1 & 0 & 2 & 0 & 2 & 2 & 3.2 \\
		34 & 0 & 0 & 2 & 0 & 2 & 2 & 1 & 8.6 & 75 & 0 & 2 & 0 & 2 & 1 & 1 & 2 & 1.4 \\
		35 & 0 & 1 & 2 & 0 & 0 & 1 & 1 & 29.9 & 76 & 1 & 0 & 1 & 2 & 1 & 2 & 2 & 1.9 \\
		36 & 0 & 2 & 2 & 0 & 1 & 0 & 1 & 10.9 & 77 & 1 & 1 & 1 & 2 & 2 & 1 & 2 & 1.6 \\
		37 & 0 & 0 & 0 & 1 & 1 & 0 & 1 & 4.7 & 78 & 1 & 2 & 1 & 2 & 0 & 0 & 2 & 2.3 \\
		38 & 0 & 1 & 0 & 1 & 2 & 2 & 1 & 3.4 & 79 & 2 & 0 & 2 & 2 & 0 & 1 & 2 & 2.6 \\
		39 & 0 & 2 & 0 & 1 & 0 & 1 & 1 & 6.6 & 80 & 2 & 1 & 2 & 2 & 1 & 0 & 2 & 18.2 \\
		40 & 1 & 0 & 1 & 1 & 0 & 2 & 1 & 4.8 & 81 & 2 & 2 & 2 & 2 & 2 & 2 & 2 & 2.6 \\
		41 & 1 & 1 & 1 & 1 & 1 & 1 & 1 & 2.9 &  &  &  &  &  &  &  &  &  \\
		\bottomrule
	\end{tabular}
\end{table}
	
	Based on the \(\text{B}^2\text{-ACNP}\) algorithm from Section \ref{sec4}, the complete \(\text{B}^2\text{-ACNP}\) for this design and its collective visualization results are as follows:
	
		\begin{figure}[!hpt]
		\centering 
		\begin{subfigure}[b]{0.48\textwidth}
			\centering
			\includegraphics[width=\textwidth]{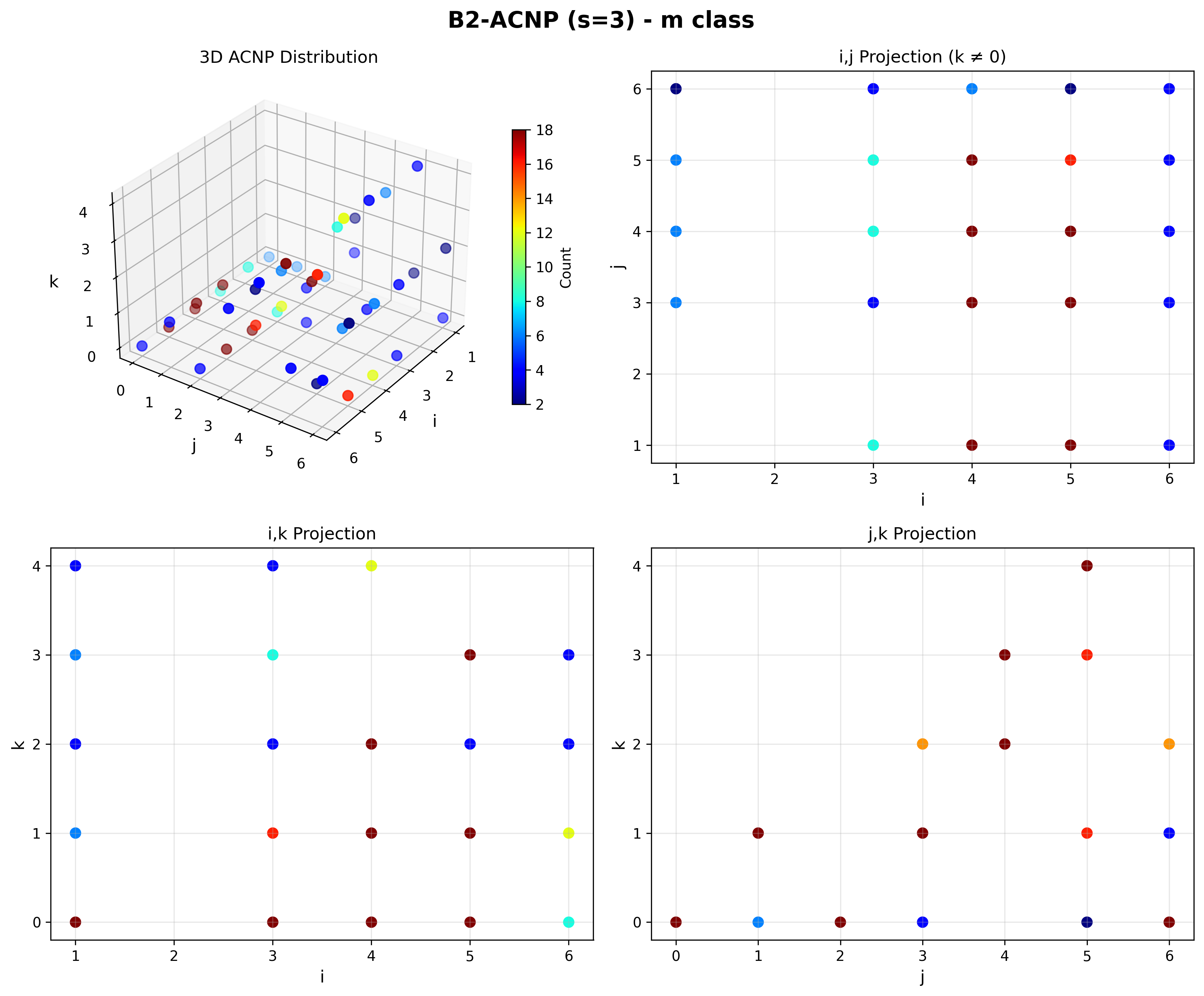} 
			\caption{$m-class$}
			\label{fig26} 
		\end{subfigure}
		\hfill
		\begin{subfigure}[b]{0.48\textwidth}
			\centering 
			\includegraphics[width=\textwidth]{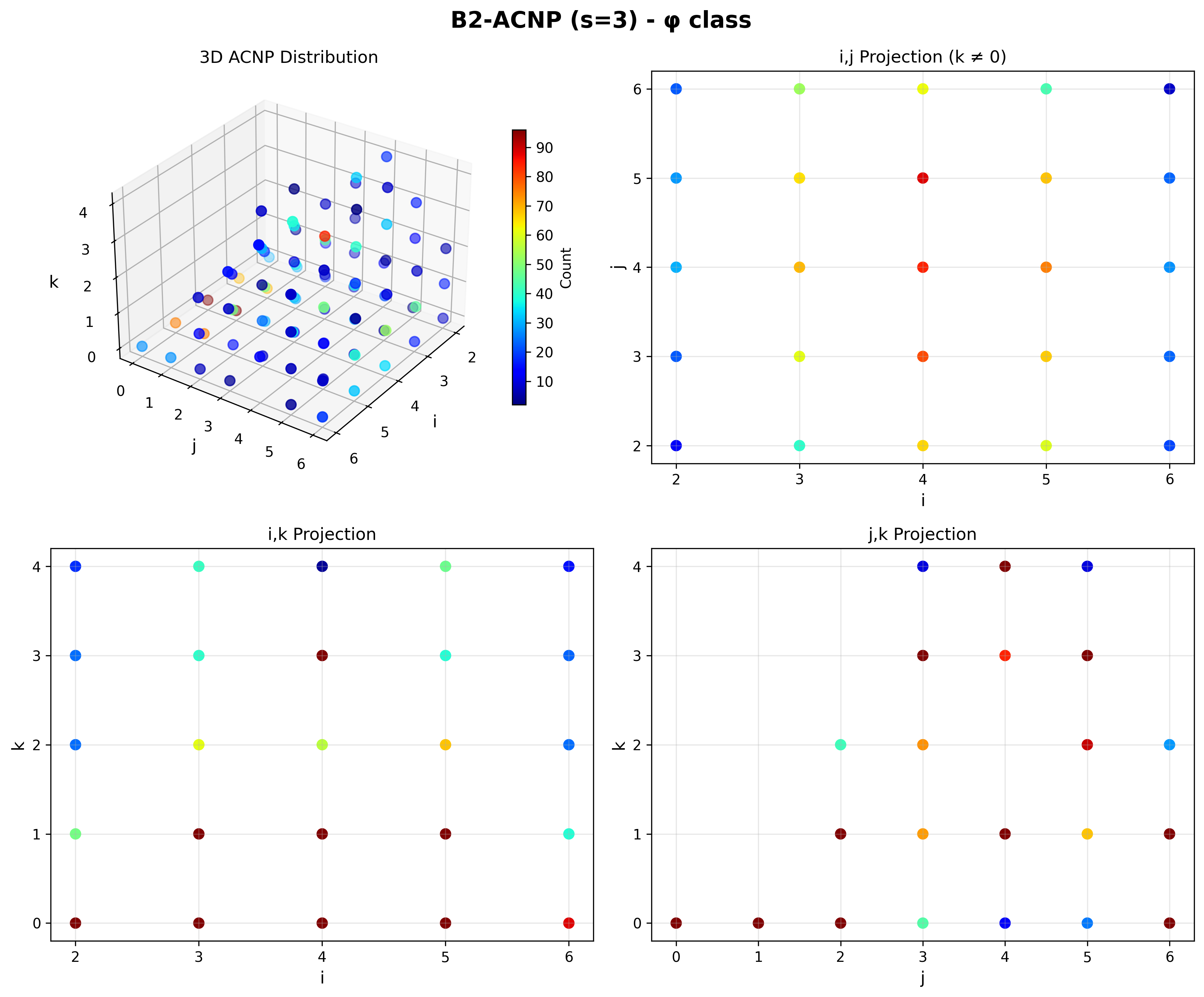} 
			\caption{ $\phi-class$ }
			\label{fig27} 
		\end{subfigure}
		\caption{Confounding pattern for the $3^{6-2}:3^1$ designs.}
		\label{fig28}
	\end{figure}

	The \(3^{6-2}:3^1\) design includes 6 main effects, 30 2fi components (2fic's), and 328 higher-order interaction components. This study focuses on the main effects, their quadratic terms, and 2fis, incorporating one block component. The following regression model is established to analyze the influence of each component:
	\[
	\begin{aligned}
	y = \beta_0 + \sum_{i=1}^6 \beta_i x_i + \sum_{i=1}^6 \beta_{ii} x_i^2 +  \sum_{1 \leq i < j \leq 6} \beta_{ij} x_i x_j + \sum_{1 \leq i < j \leq 6} \beta_{ijj} x_i x_j^2 + \gamma_1 b_1  + \varepsilon,
	\end{aligned}
	\]
	where \( y \) is the observed value, \( x_i(i=1,\ldots,6) \) represent the coded values for the six main effects of the factors, with the codes (0,1,2) transformed to (-1,0,1) for analysis. \( b_1 \) denotes the block component, \( \beta \) are the regression coefficients, and \( \epsilon \) is the random error term, assumed to be independently and identically distributed with mean zero and variance \( \sigma^2 \). 	
	Combining the aliasing information from the \(\text{B}^2\text{-ACNP}\) and the regression analysis, a model with \( R^2 = 0.816 \) was obtained below:
	\[
	\begin{aligned}
	\hat{y} &= 1.7075 - 0.0424x_1 - 0.1254x_2 - 0.0628x_3 - 1.1298x_4 - 0.2747x_5 + 0.4258x_4^2 + 0.2423x_1x_4 - 0.1600b_1.
	\end{aligned}
	\]
	

\section{Conclusion}\label{sec6}

Based on the orthogonal component system, this paper proposes a blocked aliased component-number pattern \(\text{B}^2\text{-ACNP}\) to investigate the confounding properties among component effects in $s$-level designs with multi-block variables.  Using the method of complementary sets, we prove the confounding properties of low-order components, and obtain the necessary condition for \(\text{B}^2\text{-GMC}\) designs. The explicit expressions of \(^{\#}{}_{1}^mC_2^{(k)},
^{\#}{}_{2}^b C_2^{(k)}, ^{\#}{}_{2}^\phi C_2^{(k)}\) were provided for  a type of $s^{n-m}: s^p$ MBV designs with $D_t=S_{qv}$ and $D_b\subset H_v$. To compute all elements in the \(\text{B}^2\text{-ACNP}\), a blocked wordlength distribution matrix is introduced to characterize aliasing among components within alias sets.  In addition, the classification patterns of MA, CE, and MEC criteria can be uniformly expressed by both the B-WDM  and the specific elements of the \(\text{B}^2\text{-ACNP}\). This indicates that the \(\text{B}^2\text{-GMC}\) criterion makes more comprehensive use of the inherent aliasing information in MBV designs than existing criteria with a single block. Theoretical proofs and case studies further demonstrate that the \(\text{B}^2\text{-ACNP}\) provides a more refined aliasing description than the B-ACNP and is more suitable for MBV designs than the \(\text{B}^1\text{-ACNP}\) under $Kind~2$. Moreover, we provide a Python implementation of the algorithm for computing the complete \(\text{B}^2\text{-ACNP}\) and performing visual analysis. However, some open problems remain unsolved. Future research will further focus on the systematic construction of \(s\)-level \(\text{B}^2\text{-GMC}\) designs.

	\section*{Conflict of interest}
	The authors declare that they have no conflict of interest.
	
	
\section*{Funding}
The work was supported by the National Natural Science Foundation of China (12561047),  the Xinjiang Talent Development Fund (XJRC-2025-KJ-PY-KJLJ-108), and the 2025 Central Guidance for Local Science and Technology Development Fund (ZYYD2025ZY20).

\bibliography{reference}

\section*{Appendix A: Python code for calculating the complete B$^2$-ACNP.}

\begin{python}
from collections import defaultdict
import itertools
	
def defining_contrast_subgroup(word, s, v):
    m = len(word)
    all_vectors = [list(x) for x in itertools.product(range(s), repeat=m)]
    if v == 1:
        all_vectors = [vec for vec in all_vectors if sum(1 for e in vec if e != 0) != 1]
    elif v == 2:
        all_vectors = [vec for vec in all_vectors if sum(1 for e in vec if e != 0) == 2]
    o_3 = []
    for coeffs in all_vectors:
        combination = [0] * len(word[0])
        for i in range(m):
            scaled = [coeffs[i] * elem for elem in word[i]]
            combination = [(a + b) 
        o_3.append(combination)
    generate_word = []
    for vec in o_3:
        first_non_zero = next((x for x in vec if x != 0), None)
        if first_non_zero is not None and first_non_zero != 1:
            continue
        generate_word.append(vec)
    return generate_word + word

def block_effect(block, G, s, v=2):
    block1 = defining_contrast_subgroup(block, s, v)
    unique_vectors = []
    seen = set()
    for vec in block1:
        vec_tuple = tuple(vec)
        if vec_tuple not in seen:
            seen.add(vec_tuple)
            unique_vectors.append(vec)
    block1 = unique_vectors
    block2 = []
    for b_vec in block1:
        for g_vec in G:
            combination = [(g_vec[k] + b_vec[k]) 
            block2.append(combination)
    return block2

def generate_yates(n, m, s):
    vectors = []
    q = n - m
    for i in range(s ** q):
        digits = []
        temp = i
        for _ in range(q):
            temp, rem = divmod(temp, s)
            digits.append(rem)
        first_non_zero = next((elem for elem in digits if elem != 0), None)
        if first_non_zero is not None and first_non_zero != 1:
            continue
        vectors.append(digits + [0] * m)
    return vectors

def compute_b2_acnp(d_matrix, block, m, s):
    n = len(d_matrix[0])
    beffect = block_effect(block, d_matrix, s, v=2)
    yates = generate_yates(n, m, s)
    g_class = []
    b_class = []
    m_class = []
    phi_class = []
    block_effect_set = set(tuple(vec) for vec in beffect)
    design_half = d_matrix[1:]
    extend_design = d_matrix.copy()
    for mult in range(2, s):
        for row in design_half:
            multiplied_row = [(mult * elem) 
            extend_design.append(multiplied_row)
    for i, effect in enumerate(yates):
        new_design = [[(x + y) 
        t = [sum(1 for elem in row if elem != 0) for row in new_design]
        length_of_subgroup = len(d_matrix)
        if i == 0:
            g_class.append(t[:length_of_subgroup])
        else:
            contains_block = any(tuple(row) in block_effect_set for row in new_design[:length_of_subgroup])
            if contains_block:
                b_class.append(t)
            elif 1 in t:
                m_class.append(t)
            else:
                phi_class.append(t)
    def compute_class_aenp(t_list, aenp_dict):
        for t in t_list:
            for i in range(n + 1):
                p_i = t.count(i)
                if p_i == 0:
                    continue
                for j in range(n + 1):
                    q_j = t.count(j)
                    k = q_j - 1 if i == j else q_j
                    aenp_dict[(i, j, k)] += p_i
    g_aenp = defaultdict(int)
    b_aenp = defaultdict(int)
    m_aenp = defaultdict(int)
    phi_aenp = defaultdict(int)
    compute_class_aenp(g_class, g_aenp)
    compute_class_aenp(b_class, b_aenp)
    compute_class_aenp(m_class, m_aenp)
    compute_class_aenp(phi_class, phi_aenp)
    return {
    	'g_aenp': dict(g_aenp),
    	'b_aenp': dict(b_aenp),
    	'm_aenp': dict(m_aenp),
    	'phi_aenp': dict(phi_aenp)
    }
	
def print_b2_acnp(acnp_results):
    def print_single_aenp(aenp, class_name):
        print(f"\n{class_name}-class B2-ACNP:")
        print("=" * 50)
        all_i_values = sorted(set(key[0] for key in aenp.keys()))
        all_j_values = sorted(set(key[1] for key in aenp.keys()))
        for i in all_i_values:
            for j in all_j_values:
                relevant_keys = [(i_, j_, k_) for (i_, j_, k_) in aenp.keys() if i_ == i and j_ == j]
                if not relevant_keys:
                    continue
                non_zero_keys = [k for k in relevant_keys if aenp[k] != 0]
                if not non_zero_keys:
                    continue
                max_k = max(k[2] for k in non_zero_keys)
                result_list = [aenp.get((i, j, k), 0) for k in range(max_k + 1)]
                print(f"#{i}C{j}: {result_list}")
    print_single_aenp(acnp_results['g_aenp'], 'g')
    print_single_aenp(acnp_results['b_aenp'], 'b')
    print_single_aenp(acnp_results['m_aenp'], 'm')
    print_single_aenp(acnp_results['phi_aenp'], 'phi')
\end{python}


\end{document}